\documentclass{article}
\usepackage{ijcai24}

\usepackage{times}
\usepackage{soul}
\usepackage{url}
\usepackage[hidelinks]{hyperref}
\usepackage[utf8]{inputenc}
\usepackage[small]{caption}
\usepackage{graphicx}
\usepackage{amsmath}
\usepackage{amsthm}
\usepackage{booktabs}
\usepackage{algorithm}
\usepackage{algorithmic}
\usepackage[switch]{lineno}
\usepackage{makecell}
\usepackage{enumerate}
\usepackage[inline]{enumitem}
\usepackage{amsmath}
\usepackage{amsthm}
\usepackage{amssymb}
\usepackage{graphicx}
\usepackage{xcolor}
\usepackage{multirow}
\usepackage{nicematrix}

\newcommand{\floor}[1]{\left\lfloor #1 \right\rfloor}
\newcommand{\ceil}[1]{\left\lceil #1 \right\rceil}

\newtheorem{mechanism}{Mechanism}

\def \erre{\mathbb{R}}
\def \MCo{MC}
\def \SCo{SC}

\newtheorem{example}{Example}
\newtheorem{theorem}{Theorem}

\title{Facility Location Problems with Capacity Constraints: Two Facilities and Beyond}

\author{
Gennaro Auricchio$^1$
\and
Zihe Wang$^{2}$\And
Jie Zhang$^1$
\affiliations
$^1$University of Bath, Department of Computer Science\\
$^2$Renmin University of China
\emails
ga647@bath.ac.uk,
wang.zihe@ruc.edu.cn,
jz2558@bath.ac.uk
}
\begin{document}

\maketitle

\begin{abstract}
% {\color{red} 
In this paper, we investigate the Mechanism Design aspects of the $m$-Capacitated Facility Location Problem ($m$-CFLP) on a line. We focus on two frameworks.
In the first framework, the number of facilities is arbitrary, all facilities have the same capacity, and the number of agents is equal to the total capacity of all facilities.
In the second framework, we aim to place two facilities, each with a capacity of at least half of the total agents. 
For both of these frameworks, we propose truthful mechanisms with bounded approximation ratios with respect to the Social Cost (SC) and the Maximum Cost (MC). 
When $m>2$, the result sharply contrasts with the impossibility results known for the classic $m$-Facility Location Problem \cite{fotakis2014power}, where capacity constraints are not considered.
Furthermore, all our mechanisms are
\begin{enumerate*}[label=(\roman*)]
    \item optimal with respect to the MC;
    \item  optimal or nearly optimal with respect to the SC among anonymous mechanisms.
\end{enumerate*} 
For both frameworks, we provide a lower bound on the approximation ratio that any truthful and deterministic mechanism can achieve with respect to the SC and MC. 

\end{abstract}

\section{Introduction}

Mechanism Design aims to establish procedures for aggregating the private information of a group of agents to optimize a social objective.
However, optimizing the social objective solely based on reported preferences frequently results in undesirable manipulation due to the self-interested behavior of the agents.
Therefore, one of the most crucial requirements for a mechanism is the property of \textit{truthfulness}, which ensures that no agent can gain an advantage by misreporting their private information.
Unfortunately, this strict property often conflicts with optimizing the social objective, so the output of a truthful mechanism is oftentimes sub-optimal.
To quantify the loss in efficiency, Nisan and Ronen introduced the concept of \textit{approximation ratio}. 
This quantity represents the highest achievable ratio between the social objective obtained by a truthful mechanism and the optimal social objective among all possible agents' reports \cite{nisan1999algorithmic}.
One of the classic examples of these problems is the $m$-Facility Location Problem ($m$-FLP).
In its most basic form, the $m$-FLP consists in locating $m$ facilities amongst $n$ self-interested agents.
Every agent needs to access a facility, so they would prefer to have one of the facilities placed as close as possible to their position.
Moreover, each facility can serve any number of agents, thus every mechanism just needs to return the positions of the facilities. The agents are then free to decide which facility to use, without considering any possible overload.
%

% 
% The Mechanism Design aspects of this problem have been extensively studied ever since Procaccia and Tennenholtz introduced the framework in their seminal paper \cite{procaccia2013approximate}.
% % 
% In particular, it is well-known that there is no deterministic, anonymous, and truthful mechanism with a bounded approximation ratio for $m$-FLP on the line whenever $m>2$ \cite{fotakis2014power}.
% 
% Despite these negative results, the wide applicability of this problem to fields such as disaster relief \cite{doi:10.1080/13675560701561789}, supply chain management \cite{MELO2009401}, healthcare \cite{ahmadi2017survey}, and public facilities accessibility \cite{barda1990multicriteria}; kept the interest for the $m$-FLP stand the test of time.
% 
% This interest is also supported by the thriving of alternative formulation, such as the Obnoxious FLP \cite{ibara2012characterizing,cheng2013strategy}, the Heterogeneous FLP \cite{feigenbaum2014strategyproof,zou2015facility}, and the Capacitated FLP \cite{aziz2020capacity}, which is the main topic of the present paper.
% % 

% \ga{pre-fixed capacity limit, namely $c_1,\dots,c_m$}.
% 
% \ga{In this framework, we first need to determine $m$ positions on the line. 
% % 
% Subsequently, since each facility has a different capacity, we must specify how to pair each elicited location with one of the facilities. 
% % 
% Finally, we need to allocate the agents to the facilities w the capacity limits.}
% 
% the positions of the facilities, 
% 
% then
% , since every facility has a different capacity, 
% 
In this paper, we study the $m$-Capacitated Facility Location Problem ($m$-CFLP) on the line \cite{pal2001facility}.
The $m$-CFLP is a natural extension of the $m$-FLP in which every facility has a capacity limit. 
% 
% {\color{red}
Considering facilities with capacity constraints is a natural approach for modeling scenarios where facilities offer a limited resource, as it happens in distribution planning \cite{pochet1988lot} and telecommunication network design \cite{boffey1989location,chardaire1999hierarchical}.
%
% In these scenarios, it is crucial to prevent resource waste, as it is suboptimal. 
% 
For instance, the facilities represent servers while the agents represent tasks awaiting execution, or facilities could be grocery shops and agents the customers in need of service.
% }
% 
% Considering facilities with capacity constraints is a natural way to model scenarios in which the facilities provide a limited resource as it happens in distribution planning \cite{pochet1988lot} and in telecommunication network design \cite{boffey1989location,chardaire1999hierarchical}.
% 
% In most of these scenarios, it is important to avoid any resource waste as it is suboptimal, e.g. the facilities are servers and the agents are tasks that needs to be executed, or the facilities are grocery shops and the agents are the customers.
% 
% Furthermore, in many cases, each facility comes with the same capacity, as it happens for servers, which justifies our 
% 

% 
% {\color{red}
% 
% \blue{
From a mechanism design perspective, the study of the $m$-FLP and $m$-CFLP differs significantly, allowing to elude the impossibility results known for the $m$-FLP \cite{fotakis2014power,walsh2020strategy}.
% 
% First, defining a mechanism for the $m$-CFLP is more involved, as we are need to elicit $m$ positions and then specify which facility is built at each position.
% 
% Indeed, the optimal facility location changes, allowing to bypass the impossibility results known for the $m$-FLP \cite{fotakis2014power,walsh2020strategy}.
% 
In particular, we show that this is the case when we have $m$ facilities with equal capacity $k$ and the number of agents is $n=km$.
For this class of problems, we characterize both the upper and lower bounds of the approximation ratio of anonymous, truthful, and deterministic mechanisms for $m\ge2$, showing that both are bounded, and that they coincide, making the bounds tight.
It is also noteworthy that the study of the $m$-CFLP is contingent upon the specifics of the problem.
% 
% }
% 
Indeed, the properties of any mechanism depend on factors such as whether different facilities have different capacities \cite{aziz2020facility}, whether the total capacity is larger than the number of agents \cite{ijcai2022p75}, or whether the capacity of each facility is lower than a critical threshold\cite{aziz2020facility}.
% \begin{enumerate*}[label=(\roman*)]
%     \item whether different facilities have different capacities \cite{aziz2020facility},
%     \item whether the total capacity is larger than the number of agents \cite{ijcai2022p75}, or
%     \item whether the capacity of each facility is lower than a critical threshold.
% \end{enumerate*}  
% 
For this reason, the few results providing tight lower bounds on the approximation ratio are limited to very specific settings \cite{aziz2020facility}.

\paragraph{Our Contribution.}
In this paper, we study two relevant frameworks for the $m$-CFLP from a Mechanism Design perspective.
First, we study the $m$-CFLP with equi-capacitated facilities and no spare capacity, i.e. the $m$-CFLP in which all the facilities have the same capacity, namely $k$, and the total capacity of the facilities equals the number of agents.
% 
% {\color{red} This framework describes the cases in which}
% 
We present two truthful and anonymous mechanisms, the Propagating Median Mechanism (PMM) and the Propagating InnerPoint Mechanism (PIPM).
We show that both the PMM and the PIPM have a bounded approximation ratio with respect to the Social Cost (SC) and the Maximum Cost (MC), regardless of the value of $m$.
This result stands in contrast with the classic results for the $m$-FLP, according to which no mechanism can be deterministic, anonymous, truthful, and achieve a finite approximation ratio when $m>2$, even on the line \cite{fotakis2014power,walsh2020strategy}.
We then present three lower bounds for the approximation ratio for the $m$-CFLP with equi-capacitated facilities and no spare capacities.
In particular:
\begin{enumerate*}[label=(\roman*)]
    \item no truthful and deterministic mechanism achieves an approximation ratio lower than $2$ with respect to MC. Thus, PMM and PIPM are optimal with respect to this metric.
    \item When $k>3$, no truthful and deterministic mechanism can achieve an approximation ratio lower than $3$ with respect to the SC.
    \item No truthful, deterministic, and anonymous mechanism achieves an approximation ratio with respect to SC lower than $\big(\frac{k(m-1)}{2}+1\big)$, if $m$ is odd, or lower than $\big(\frac{km}{2}-1\big)$ if $m$ is even. In particular, the PMM and the PIPM are the best possible truthful, deterministic, and anonymous mechanisms for odd and even $m$, respectively.
\end{enumerate*}

% {\color{red}
% 
We then study the $2$-CFLP with abundant facilities, in which we have two facilities capable to accommodate at least half of the agents.
This framework has been studied under further assumptions: \begin{enumerate*}[label=(\roman*)]
    \item in \cite{aziz2020facility}, the authors studied the case in which $n$ is even, $c_1=c_2=\frac{n}{2}$, and proposed the InnerPoint (IM) Mechanism, 
    \item in \cite{ijcai2022p75}, the author studied the case in which $n$ is odd and $c_1=\ceil{\frac{n}{2}}$, $c_2=\floor{\frac{n}{2}}$, and proposed the InnerChoice (IC) Mechanisms introduced,
    \item again, in \cite{ijcai2022p75}, the author also studied the case in which $n$ is arbitrary but $c_1=c_2$, and proposed the InnerGap (IG) Mechanism.
\end{enumerate*}
However, this is the first time that a study of a framework encompassing all these different cases has been conducted.
We propose the Extended InnerGap (EIG) Mechanism, which generalizes and includes IM, IC, and IG.
The EIG is strong Group Strategyproof, thus truthful, and attains a bounded approximation ratio with respect to the SC and MC.
We then provide a lower bound on the approximation ratio of any truthful mechanism with respect to the SC and MC and show that the EIG is optimal with respect to the MC.
Moreover, the EIG is optimal with respect to the SC whenever $n\ge \bar c+\sqrt{\bar c}$, where $\bar c=\max\{c_1,c_2\}$.
In Table \ref{tab:my_table}, we summarize our findings in terms of lower and upper bounds for the cases we study.
Due to space limits, some proofs are deferred to the Appendix.
% }

\begin{table*}[t]
  \centering
  \begin{NiceTabular}{@{}c@{\hskip 0.3in} | c c c@{\hskip 0.3in} | c c@{}}
    \toprule
     &   \multicolumn{3}{c}{\textbf{Social Cost}} &   \multicolumn{2}{c}{\textbf{Maximum Cost}}  \\
    % \hline
     &  \textbf{LB} &  \textbf{LB*} & \textbf{UB} &  \textbf{LB}  & \textbf{UB} \\
    % \midrule
    \hline
    \\[-0.7em]
    \multirow{2}{*}{\makecell{$c_{j}=k$\\$n=km$}} & \multirow{2}{*}{$3$} &  \multirow{2}{*}{\makecell{$\frac{k(m-1)}{2}+1$ \quad ($m$ odd)\\$\frac{km}{2}-1$\quad\quad\quad($m$ even)}} &  \multirow{2}{*}{\makecell{$\frac{k(m-1)}{2}+1$ \quad ($m$ odd)\\$\frac{km}{2}-1$\quad\quad\quad($m$ even)}}  & \multirow{2}{*}{2} & \multirow{2}{*}{2}\\
    %     % \cline{2-5}
        &  & & & &  \\[0.2cm]
    % \midrule
    \hline
    \\[-0.7em]
    % % \multirow{3}{2cm}{$c_1=k$,\\[.4\baselineskip]
    % % $c_2=k+1$\\[.4\baselineskip]
    % % $n=2k+1$} & \multirow{2}{*}{\makecell{$c_{j}=k$\\$n=km$}} &  \multirow{2}{*}{} &  \multirow{2}{*}{} &  \multirow{2}{*}{} \\
    % %  \cline{2-5}
    % %  &  MC & 2 &  2 &  2  \\
    % \hline
    \multirow{2}{*}{\makecell{$c_1,c_2\ge\floor{\frac{n}{2}}$\\$c_1+c_2\ge n$}} & \multirow{2}{*}{3} &   \multirow{2}{*}{$n-\bar c -1$} &  \multirow{2}{*}{$\makecell{\max\{n-\bar c,\frac{\bar c}{n-\bar c}\}-1}$} & \multirow{2}{*}{2} & \multirow{2}{*}{2} \\
    &   &  &   &  &  \\
    % 
    % \midrule
    % % % 
    % \multirow{2}{*}{\makecell{$c_{1}=k+1,c_2=k$\\$n=2k+1$}} & \multirow{2}{*}{3}  & \multirow{2}{*}{$k-1$} &  \multirow{2}{*}{$k-1$} & \multirow{2}{*}{2} & \multirow{2}{*}{2}\\
    %     &   &  &   &  &  \\
    %     % \cline{2-5}
    %     % 
        % &   &  &  &  &  \\
    \bottomrule
  \end{NiceTabular}
  \caption{Each row contains the Lower and Upper Bounds (LB and UB, respectively) with respect to the Social and Maximum Cost for a different class of problems for the $m$-CFLP. The value $\bar c$ is the maximum capacity of the facilities. The LB column contains the lower bounds for the class of truthful and deterministic mechanisms. The $LB^*$ column contains the lower bounds for the class of mechanisms that are truthful, deterministic, and anonymous. 
  % \blue{better to add a column line to separate SC and MC.}
  }
  \label{tab:my_table}
\end{table*}

% \begin{table}[]
%     \centering
%     \begin{tabular}{ @{}c c c c@{}}
%     \toprule
%           & $\forall n\in\mathbb{N}$ & $n<c_1+c_2$ &$c_1\neq c_2$ \\
          
%          \midrule

%           $EIG$ & Yes & Yes & Yes \\
%           $EG$ & Yes & Yes & No \\
%           $IC$ & No & No & Yes \\
%           $IM$ & No & No & No \\
%           \bottomrule
%     \end{tabular}
%     \caption{Frameworks under which the mechanisms operate when $c_1,c_2\ge \floor{\frac{n}{2}}$. From right to left, the column tell us whether the mechanism is capable of working ($i$) for every number of agents $n$, ($ii$) when the total capacity is larger than the number of agents, and ($iii$) when the two facilities have different capacities.
%     % 
%     The EIG (Exended InnerGap) Mechanism is the only mechanism capable of working under no further restriction.}
%     \label{tab:my_label}
% \end{table}

\paragraph{Related works.}
The $m$-Facility Location Problem ($m$-FLP) and its variants are relevant problems in several applied fields such as disaster relief \cite{doi:10.1080/13675560701561789}, supply chain management \cite{MELO2009401}, healthcare \cite{ahmadi2017survey}, clustering \cite{hastie2009elements}, and public facilities accessibility \cite{barda1990multicriteria}.
The Mechanism Design study of the $m$-FLP was first explored by Procaccia and Tennenholtz, who laid the foundation of this field in their pioneering work \cite{procaccia2013approximate}.
Subsequently, several mechanisms with small constant approximation ratios for locating one or two facilities on trees, circles, and general graphs have been introduced \cite{10.2307/40800845,DBLP:conf/sigecom/DokowFMN12,DBLP:conf/sigecom/FeldmanW13,DBLP:conf/atal/FilimonovM21,DBLP:conf/sigecom/LuSWZ10,DBLP:conf/wine/LuWZ09,DBLP:journals/aamas/Filos-RatsikasL17,DBLP:conf/sagt/Meir19,DBLP:conf/sigecom/TangYZ20}.
However, these positive results are limited to cases where the mechanism designer needs to place $2$ facilities or the mechanism is not deterministic. 
Indeed, no deterministic, anonymous, and truthful mechanism can place more than two uncapacitated facilities while achieving a bounded approximation ratio even on the line \cite{fotakis2014power,walsh2020strategy}.
The $m$-Capacitated Facility Location Problem ($m$-CFLP) is a natural extension of the $m$-FLP, in which each facility has a maximum number of agents that it can serve \cite{brimberg2001capacitated,pal2001facility,aardal2015approximation}.
The Mechanism Design aspects of the $m$-CFLP have received relatively little attention until recently.
Indeed, the game theoretical framework for the $m$-CFLP that we consider was first introduced in \cite{aziz2020facility}.
In this paper, the authors studied various truthful mechanisms (such as the InnerPoint Mechanism and the Extended Endpoint Mechanism) and studied their approximation ratios.
Notably, only mechanisms capable of locating two facilities achieve a bounded approximation ratio.
A more theoretical analysis of the problem has been then presented in \cite{ijcai2022p75}, where the author demonstrated that no mechanism can locate more than two capacitated facilities while being truthful, anonymous, and Pareto optimal. 
Lastly, papers that deal with different Mechanism Design aspects of the $m$-CFLP are \cite{auricchio2023extended}, where the $m$-CFLP is studied in a Bayesian setting, and \cite{aziz2020capacity}, where the authors investigate the case in which there is only one capacitated facility to place and it cannot accommodate all the agents.

\section{Preliminaries}
% {\color{red}
% 
In this section, we introduce the two $m$-CFLP frameworks and fix the notation.
Throughout the paper, we assume that the agents lay on a line and denote with $\vec x :=(x_1,\dots,x_n)\in \erre^n$ the vector containing their positions.
Moreover, we assume $m>1$, as $1$-CFLP is equivalent to the classic $1$-FLP.

\paragraph*{The $m$-CFLP with equi-capacitated facilities and no spare capacity.}
In the first framework, we have $m$ facilities whose capacity is the same, namely $c_j=k$ for every $j\in[m]$, and the total capacity of the facilities equals the total number of agents, hence $n=mk$.
We call this framework the $m$-CFLP with \textit{equi-capacitated facilities} and \textit{no spare capacity}.
Since in this setting all the facilities have the same capacity, a facility location is defined by two objects:
\begin{enumerate*}[label=(\roman*)]
    \item a $m$-dimensional vector $\vec y=(y_1,\dots,y_m)$ whose entries are the positions of the facilities on the line, and
    \item a matching $\mu\subset [n]\times [m]$ that determines how the agents are assigned to facilities, i.e. $(i,j)\in\mu$ if and only if the agent at $x_i$ is assigned to $y_j$. Due to the capacity constraints, the degree of every vertex $j\in[m]$ according to $\mu$ is at most $k$. Since every agent is assigned to only one facility, the degree of $i\in[n]$ according to $\mu$ is $1$.
\end{enumerate*} 
% 
% Notice that, since every facility has the same capacity, we do not need to specify

\paragraph*{The $2$-CFLP with abundant facilities.}
In the second framework we consider, we have two facilities whose capacities, namely $c_1$ and $c_2$ are such that $\floor{\frac{n}{2}}\le c_2, c_1\le n-1$.
We call this framework $2$-CFLP with \textit{abundant capacities}.
Since the facilities may have different capacity, eliciting two positions, namely $y_1$ and $y_2$, and an agent-to-facility assignment $\mu$ is not sufficient, as we need to also specify the capacity of the two facilities.
In particular, in this framework, a facility location is defined by three objects:
\begin{enumerate*}[label=(\roman*)]
    \item a bi-dimensional vector $\vec y=(y_1,y_2)$ whose entries are the positions of the facilities,
    \item a permutation $\pi:[2]\to[2]$ that specifies the capacity of each facility, so that if $\pi(1)=j$ the facility at $y_1$ has capacity $c_j$ and the facility built at $y_2$ has capacity $c_i$ with $i\neq j$ and $i,j\in[2]$, and
    \item a matching $\mu\subset [n]\times [2]$ that determines how the agents are assigned to facilities. The degree of every vertex $j\in[2]$ according to $\mu$ must be at most $c_{\pi(j)}$, while the degree of every $i\in[n]$ according to $\mu$ is $1$.
\end{enumerate*}

\paragraph{Mechanism Design Framework for the $m$-CFLP.}
In both frameworks, given the positions of the facilities $\vec y$ and a matching $\mu$, we define the cost of an agent positioned in $x_i$ as $c_{i,\mu}(x_i,\vec y)=|x_i-y_j|$, where $(i,j)$ is the unique edge in $\mu$ adjacent to $i$.
Finally, a cost function is a map $C_\mu:\erre^n\times \erre^m\to [0,+\infty)$ that associates to $(\vec x , \vec y)$ the overall cost of placing the facilities at $\vec y$ and assigning the agents positioned at $\vec x$ according to $\mu$.\footnote{In what follows, we omit $\mu$ from the indexes of $c$ and $C$ if it is clear from the context which matching we are considering.}
For both frameworks, given a vector $\vec x\in\erre^n$ containing the agents' positions, the $m$-\textit{Capacitated Facility Location Problem} with respect to the cost $C$, consists in finding the locations for $m$ facilities and a matching $\mu$ that minimize the function $\vec y \to C(\vec x,\vec y)$.
Throughout the paper, we consider the \textit{Social Cost} ($\SCo$), defined as the sum of all the agents' costs, i.e., $\SCo(\vec x, \vec y)=\sum_{i\in [n]}c_i(x_i,\vec y)=\sum_{i\in[n]}|x_i- y_j|$ and the \textit{Maximum Cost} ($\MCo$), defined as the maximum cost among all agents' costs, i.e. $\MCo(\vec x, \vec y):=\max_{i\in [n]}c_i(x_i,\vec y)$. 
A mechanism for the $m$-CFLP is a function $f$ that takes the private information of $n$ self-interested agents as input and returns a facility location.
Thus, for the $m$-CFLP with equi-capacitated facilities and no spare capacity, the mechanism returns a set of locations $\vec y$ and a matching $\mu$ between the agents and the facilities.
%
% Every agent is self-interested, thus it will misreport its position if this lowers its cost.
% 
A mechanism $f$ is said to be \textit{truthful} (or \textit{strategy-proof}) if, for every agent, its cost is minimized when it reports its true position, i.e., for any $x_i'\in \erre$, we have $c_i(x_i,f(\vec x))\le c_i(x_i,f(\vec x_{-i},x_i'))$ where $x_i$ is the agent's real position and $\vec x_{-i}$ is the vector $\vec x$ without its $i$-th component.
% 
% Another important property for a mechanism is \textit{strong Group Strategyproofness}.
% 
A mechanism $f$ is strong Group Strategyproof (GSP) if no group of agents can misreport their positions in such a way that
\begin{enumerate*}[label=(\roman*)]
    \item the cost of every agent in the group after manipulating is less than or equal to the cost they would get by reporting truthfully,
    \item at least one of the agents in the group incurs a strictly lower cost after the group manipulation.
    % \footnote{This property is also known as Strong Group Strategyproofness, however, for ease of notation, we denote it with GSP.}
\end{enumerate*}

Albeit a truthful mechanism prevents agents from misreporting their positions, their output is usually suboptimal.
To evaluate this efficiency loss, we consider the approximation ratio of the mechanism introduced in \cite{nisan1999algorithmic}.
Given a truthful mechanism $f$, its approximation ratio with respect to the SC is defined as $ar_{SC}(f):=\sup_{\vec x \in \mathbb{R}^{n}}\frac{SC_{f}(\vec x)}{SC_{opt}(\vec x)}$, where $SC_f(\vec x)$ is the SC of the solution returned by $f$ and $SC_{opt}(\vec x)$ is the optimal SC achievable on instance $\vec x$.
Similarly, the approximation ratio of $f$ with respect to the MC (namely, $ar_{MC}(f)$) is the highest ratio between the MC achieved by $f$ and the optimal MC.
% 
% }

\section{The \texorpdfstring{$m$}{m}-CFLP with Equi-capacitated Facilities and no Spare Capacity}
\label{sec:beyond}

In this section, we focus on the Mechanism Design aspects of the $m$-CFLP with equi-capacitated facilities and no spare capacity, i.e. given $m$ the number of facilities and their capacity $k$, it holds $n=mk$.
% 
% We assume that $k>2$, as the case $k=1$ is trivial.
% 
Given $j\in [m]$ and a vector $\vec{x}$ containing the agents' positions ordered from left to right, i.e. $x_i\le x_{i+1}$ for every $i\in [n-1]$, we define $I_j$ as follows
\begin{equation}
\label{eq:jcluster}
I_j=\{x_{(j-1)k+i}\;\;\text{where }\; i\in[k]\}.
\end{equation}
Notice that, since every facility has the same capacity, the optimal solution to the $m$-CFLP with respect to the SC places the facilities at the positions $\vec{y}=(y_1,\dots,y_m)$, where each $y_j$ is a median of set $I_j$, and then it assigns every agent whose position is in $I_j$ to $y_j$.
Likewise, the optimal solution with respect to the MC places the facilities at $y_j=\frac{x_{(j-1)k+1}+x_{jk}}{2}$ and then it assigns the agents in $I_j$ to $y_j$.

\subsection{The Propagating Median Mechanism}

We now introduce and study our first truthful mechanism for the $m$-CFLP with equi-capacitated facilities and no spare capacity, the Propagating Median Mechanism (PMM).

\begin{mechanism}[Propagating Median Mechanism (PMM)]
    Let $n$ be the total number of agents, $m$ the number of facilities to place, and $k=\frac{n}{m}\in \mathbb{N}$ the capacity of each facility.
    Let us set $r=\floor{\frac{m+1}{2}}$.
    The routine of the PMM is as follows:
    \begin{enumerate*}[label=(\roman*)]
        \item First, we locate the facility $y_r$ at one of the medians of $I_r$, i.e. $y_r=x_{k(r-1)+\floor{\frac{k+1}{2}}}$.
        \item Second, we determine the positions of the other facilities via the following iterative routine. 
        For any $l\ge r$, let us be given the position of the $l$-th facility, namely $y_l$. Then the position of the $(l+1)$-th facility is $y_{l+1}=\max\{x_{kl+1},x_{kl}+d_l\}$, where $d_l$ is the distance between $y_{l}$ and $x_{kl}:=\max_{x\in I_{l}}\; x $.
        Similarly, given $l\le r$ and the position of the $l$-th facility, namely $y_l$, the $(l-1)$-th facility is placed at $y_{l-1}=\min\{x_{k(l-1)},x_{k(l-1)+1}-d_l\}$, where $d_l$ is the distance between $y_l$ and $x_{k(l-1)+1}:=\min_{x\in I_{l}}\; x$.
        \item Finally, all the agents in $I_j$ are assigned to $y_j$.
    \end{enumerate*}
    
\end{mechanism}

% 
% Notice that, when $m=2$, the PMM is the InnerPoint mechanism presented in \cite{aziz2020facility}.
% 

% 
Let $\vec y=(y_1,\dots,y_m)$ be the position of the facilities returned by the PMM on a given instance $\vec x$.
It is easy to see that the entries of $\vec y$ are non-decreasing, i.e. $y_j\le y_{j+1}$ for every $j\in [m-1]$.
Moreover, the PMM assigns every agent to its closest facility, so that $\min_{\ell\in [m]}|x_i-y_\ell|=|x_i-y_j|$, where $j\in[m]$ is the only index for which it holds $x_i\in I_j$.
% 
% We now start studying the social properties of the PMM.

\begin{theorem}
\label{thm:PMM_tr}
    The PMM is truthful.
\end{theorem}

\begin{proof}
    Toward a contradiction, let $x_i$ be the real position of an agent able to manipulate by reporting $x_i'$ instead of its real position $x_i$.
    We denote with $\vec y$ and $I_j$ the positions of the facilities returned by the PMM and the sets in \eqref{eq:jcluster} on the truthful input, respectively.
    Similarly, we denote with $\vec y\,'$ and $I_j'$ the positions of the facilities returned by the PMM and the sets in \eqref{eq:jcluster} when the agent at $x_i$ reports $x_i'$, respectively.
    Since the other case is symmetric, we assume that $y_{r}\le x_i$, where $r=\floor{\frac{m+1}{2}}$ and $y_r=x_{k(r-1)+\floor{\frac{k+1}{2}}}$, is a median of $I_r$.
    We show that no agent in $I_{r}$ can manipulate, the case in which $x_i\notin I_r$ is similar and deferred to the Appendix.
    Toward a contradiction, let us assume that $x_i\in I_{r}$.
    Then, if $x_i=y_{r}$, the cost of the agent is null, thus it cannot benefit by misreporting.
    Therefore, it must be that $y_{r}<x_i$.
    If $x_i'<y_{r}$, we have that $x_i'\in I_\ell'$ with $\ell\le r$.
    In this case, we have that $y_r'\le y_r$, thus $y_l'\le\dots\le y_{r}'\le y_{r} < x_i$, which means that the manipulating agent is assigned to a facility that is not closer than $y_{r}$, hence its cost does not decrease after the manipulation.
    Finally, let us consider the case $y_r<x_i'$.
    In this case, we have that $y'_{r}=y_{r}$, since the median of $I_r$ is the same regardless of whether the manipulating agent reports truthfully or not.
    Thus, if $x_i'\in I_r'$, it will still be assigned to $y_{r}'=y_{r}$, which brings no benefit to the manipulative agent.
    So it must be that $x_i'\notin I_{r}'$, then $x_i'\in I_\ell'$, where $\ell>r$, thus the manipulating agent is assigned to $y_{\ell}'\ge y_{r+1}'$, since $\ell>r$.
    Let us denote with $x_{rk}'$ the position of the $(rk)$-th agent from the left in the manipulated instance $(x_i',x_{-i})$.
    Since $x_i'>x_i$, we have $x_{rk}'\ge x_{rk}\ge x_i$.
    We have that $y_{r+1}'=\max\{x_{rk+1}', x_{rk}'+|y_r'-x_{rk}'|\}\ge x_{rk}'+|y_{r}'-x_{rk}'| \ge x_i+|y_{r}-x_{rk}'| \ge x_i+|y_{r}-x_{i}|$, thus the cost of being assigned to $y_\ell' $ is no less than the cost of being assigned to $y_{r}$. 
\end{proof}

Although the PMM is truthful, it is not strong GSP, as the following example shows.

\begin{example}
\label{ex:noGSP}
    Let us fix $n=9$, $m=3$, and $k=3$.
    Let us consider the following instance: $x_1=x_2=x_3=0$, $x_4=x_5=1$, $x_6=2$, $x_7=2.5$, and $x_8=x_9=4$.
    The PMM places the facilities at $y_1=0$, $y_2=1$, and $y_3=3$.
    In this instance, the agents at $x_6$ and $x_7$ can collude: indeed, if $x_6$ reports $x_6'=1$, the PMM places the facilities at $y_1=0$, $y_2=1$, and $y_3=2.5$, thus the cost of $x_7$ decreases.
\end{example}

To conclude, we provide an analysis of the approximation ratio of the PMM with respect to the SC and MC.
In particular, we prove that $ar_{SC}(\textrm{PMM})$ and $ar_{MC}(\textrm{PMM})$ are finite.

\begin{theorem}
\label{thm:PMM_SC}
    It holds $ar_{SC}(\textrm{PMM})=k\floor{\frac{m}{2}}+1$.
\end{theorem}

\begin{proof}
    Let $r=\floor{\frac{m+1}{2}}$ and $\vec x$ be a vector containing all the agents' reports ordered from left to right.
    Let $SC_{opt}(\vec x)$ be the optimal SC for $\vec x$, then it holds that $SC_{opt}(\vec x)=\sum_{j\in [m]}SC_{opt}(I_j)$, where $SC_{opt}(I_j)$ is the SC of the agents whose report is in $I_j$ according to the optimal solution.
    Similarly, let $SC_{\textrm{PMM}}(\vec x)$ be the SC of instance $\vec x$ according to the output of PMM, then it holds that $SC_{\textrm{PMM}}(\vec x)=\sum_{j\in [m]}SC_{\textrm{PMM}}(I_j)$, where $SC_{\textrm{PMM}}(I_j)$ is the SC of the agents whose report is in $I_j$ according to the output of PMM.
    We denote with $\vec y$ the vector containing the locations of the facilities returned by the PMM and define $J\subset [m]$ as the set of indexes $j\in [m]$ such that $y_j\in [x_{(j-1)k+1},x_{jk}]$.
    We notice that $J$ is non-empty since $r\in J$ by definition of the PMM.
    We notice that $SC_{opt}(I_r)=SC_{\textrm{PMM}}(I_r)$.
    Let us now consider $j\in J$ such that $j\neq r$.
    Since $y_j\in [x_{(j-1)k+1},x_{jk}]$, we have that $SC_{\textrm{PMM}}(I_j)\le (k-1)|x_{(j-1)k+1}-x_{jk}|\le (k-1)SC_{opt}(I_j)$.
    Let us now consider $j\notin J$ and, without loss of generality, let us assume that $r<j$ since the other case is symmetric.
    Since $j\notin J$, there exists an index $\ell\in J$ such that $y_j=x_{k\ell}+|y_\ell-x_{k\ell}|$.
    If $\ell\neq r$, we have that $SC_{\textrm{PMM}}(I_j)\le k|y_\ell-x_{k\ell}|\le k SC_{opt}(I_\ell)$ since, for every $x_t\in I_j$, we have $x_{k\ell}\le x_t\le y_j$.
    Similarly, if $\ell=r$, we have that $SC_{\textrm{PMM}}(I_j)\le k|y_r-x_{kr}|$.
    Therefore, it holds 
      \begin{align*}
        SC_{\textrm{PMM}}&(\vec x)\le \sum_{J\ni j\neq r}(k\lambda_j +k-1) SC_{opt}(I_j)+SC_{opt}(I_r)\\
        &\quad\quad+k\gamma_r|y_r-x_{kr}|+k\gamma_l|y_r-x_{k(r-1)+1}|.
    \end{align*}
    % \begin{align*}
    %     SC_{\textrm{PMM}}&(\vec x)\le \sum_{j\in J,j\neq r}(k\lambda_j +k-1) SC_{opt}(I_j)\\
    %     &\quad\quad\quad\quad+SC_{opt}(I_r)+k\gamma_r|y_r-x_{kr}|\\
    %     &\quad\quad\quad\quad+k\gamma_l|y_r-x_{k(r-1)+1}|.
    %     % &\quad\le \sum_{j\in J,j\neq r}(k\lambda_j +k-1) SC_{opt}(I_j)\\
    %     % &\quad\quad\quad\quad+(k\Gamma+1) SC_{opt}(I_r)
    % \end{align*}
    % 
    Hence $SC_{\textrm{PMM}}\le  \sum_{j\in J,j\neq r}(k\lambda_j +k-1) SC_{opt}(I_j)+(k\Gamma+1) SC_{opt}(I_r)$ where
    \begin{enumerate*}[label=(\roman*)]
     \item $\lambda_j$, for $j\in J$ and $j\neq r$, is the number of $i\in [m]$ such that $i\neq j$ and $y_i=x_{jk}+|y_j-x_{jk}|$ if $j>r$ and the number of $i\in [m]$ such that $i\neq j$ and $y_i=x_{jk}-|y_j-x_{jk}|$ if $j<r$,
     \item $\gamma_l$ is the number of $i\in [m]$ such that $i < r$ and $y_i=x_{k(r-1)+1}-|y_r-x_{k(r-1)+1}|$,
     \item  $\gamma_r$ is the number of $i\in [m]$ such that $i > r$ and $y_i=x_{kr}+|y_r-x_{kr}|$, and
     \item  $\Gamma=\max\{\gamma_r,\gamma_l\}$.
    \end{enumerate*}   
    If we set $K_r=(k\Gamma+1)$ and $K_j=(k\lambda_j +k-1)$ for every $j\in J$ such that $j\neq r$, we have $ar_{SC}(\textrm{PMM})\le \frac{\sum_{j\in J}K_j SC_{opt}(I_j)}{\sum_{j\in J}SC_{opt}(I_j)}$,
    % \begin{align*}
    %     ar_{SC}(\textrm{PMM})
    %     % &\le \frac{\sum_{j\in J}K_j SC_{opt}(I_j)}{\sum_{j\in [m]}SC_{opt}(I_j)}\\
    %     &\le \frac{\sum_{j\in J}K_j SC_{opt}(I_j)}{\sum_{j\in J}SC_{opt}(I_j)},
    %     % \\
    %     % &\le \max_{j\in J}\{K_j\}.
    % \end{align*}
    % since $J\subset [m]$,
    thus $ar_{SC}(\textrm{PMM})\le \max_{j\in J}\{K_j\}$.
    % $\sum_{j\in J}SC_{opt}(I_j)\le \sum_{j\in [m]}SC_{opt}(I_j)$, we have $ar_{SC}(\textrm{PMM})\le \max_{j\in J}\{K_j\}$.
    % 
    Since, $\lambda_j\le \floor{\frac{m}{2}}-1$ and $\Gamma\le\floor{\frac{m}{2}}$, we have $K_j \le k\floor{\frac{m}{2}}+1$
    % \le n-(kr-1)=k\floor{\frac{m}{2}}+1$ 
    for every $j\in J$, therefore $ar_{SC}(\textrm{PMM})\le k\floor{\frac{m}{2}}+1$.

    % Moreover, we have that $K_j \le n-(kr-1)=k\floor{\frac{m}{2}}+1$ for every $j\in J$, thus we conclude that $ar_{SC}(\textrm{PMM})\le k\floor{\frac{m}{2}}+1$.
    %    
    
    % 
    Finally, to prove that $ar(\textrm{PMM})=\big(k\floor{\frac{m}{2}}+1\big)$, consider the following instance: $x_1=\dots=x_{kr-1}=0$ and $x_{kr}=\dots=x_n=1$.
    The optimal cost of this instance is $1$.
    The PMM places the facilities as it follows $y_1=\dots=y_r=0$ and $y_{r+1}=\dots=y_m=2$, thus the Social Cost of the mechanism is $n-(k\floor{\frac{m+1}{2}}-1)=k\floor{\frac{m}{2}}+1$.
\end{proof}

Through a similar argument, we retrieve the approximation ratio of PMM with respect to the MC.

% Similarly, we retrieve $ar_{MC}(PMM)$.

\begin{theorem}
\label{thm:PMM_MC}
    It holds $ar_{MC}(\textrm{PMM})=2$.
\end{theorem}

\subsection{The Propagating InnerPoint Mechanism}

We now present our second truthful mechanism for the $m$-CFLP with equi-capacitated facilities and no spare capacity, the Propagating InnerPoint Mechanism (PIPM).
The routine of the PIPM is similar to the routine of the PMM, the main difference lies in how it determines the initial facilities.
Indeed, the PMM places a facility at the median of $I_{\floor{\frac{m+1}{2}}}$, while the PIPM places two facilities: one at the maximum value of $I_{\floor{\frac{m}{2}}}$ and one at the minimum value of $I_{\floor{\frac{m}{2}}+1}$.
% position of the rightmost agent in $I_{\floor{\frac{m}{2}}}$ and one at the position of the leftmost agent in $I_{\floor{\frac{m}{2}}+1}$.
% 

\begin{mechanism}[Propagating InnerPoint Mechanism]
\label{mech_m3}
    Let $n$ be the total number of agents, $m$ be the number of facilities to place, and $k=\frac{n}{m}\in \mathbb{N}$ be the capacity of each facility.
    Let us set $r=\floor{\frac{m}{2}}$.
    The mechanism runs as follows:
    \begin{enumerate*}[label=(\roman*)]
        \item First, we locate the facilities $y_r$ and $y_{r+1}$ at the positions $x_{rk}$ and $x_{rk+1}$, respectively.
        \item To place the other facilities, we run an iterative routine. 
        For $l\ge r+1$, given the position of the $l$-th facility, namely $y_l$, we place the $(l+1)$-th facility at the position $y_{l+1}=\max\{x_{kl+1},x_{kl}+d_l\}$, where $d_l$ is the distance between $y_{l}$ and $x_{kl}$.
        For $l\le r$, we run a similar iterative routine.
        Given the position of the $l$-th facility $y_l$, the $(l-1)$-th facility is placed at $\min\{x_{k(l-1)},x_{k(l-1)+1}-d_l\}$, where $d_l$ is the distance between $y_l$ and $x_{k(l-1)+1}$.
        \item Finally, all the agents in $I_j$ are assigned to $y_j$.
    \end{enumerate*}
    
\end{mechanism}

Due to the similarities between the definition of the PMM and the PIPM, it is possible to adapt the arguments used in the proof of Theorem \ref{thm:PMM_tr}, \ref{thm:PMM_SC}, and \ref{thm:PMM_MC} to this mechanism.
In particular, the PIPM is truthful and achieves a bounded approximation ratio with respect to both the SC and MC.

\begin{theorem}
\label{thm:PIPM_all}
    The PIPM is truthful.
    Moreover, we have that $ar_{SC}(\textrm{PIPM})=k\ceil{\frac{m}{2}}-1$ and $ar_{MC}(\textrm{PIPM})=2$.
\end{theorem}

Albeit $ar_{MC}(\textrm{PMM})=ar_{MC}(\textrm{PIPM})$, the approximation ratios of the two mechanisms with respect to the SC are different.
Indeed, we have that $ar_{SC}(\textrm{PMM})<ar_{SC}(\textrm{PIPM})$ when $m$ is odd and, vice-versa, $ar_{SC}(\textrm{PIPM})<ar_{SC}(\textrm{PMM})$ when $m$ is even.
Finally, it is easy to adapt Example \ref{ex:noGSP} to show that PIPM is not strong GSP.

\subsection{Lower Bounds for the Approximation Ratio}

To conclude the section, we study the lower bounds for the approximation ratio of truthful mechanisms for the $m$-CFLP with equi-capacitated facilities and no spare capacity. 
First, we show that $2$ is the best approximation ratio for any truthful and deterministic mechanism with respect to the MC.

\begin{theorem}
\label{thm:lowerboundMC}
    No deterministic truthful mechanism for the $m$-CFLP with equi-capacitated facilities and no spare capacity can achieve an approximation ratio with respect to the Maximum Cost that is lower than $2$.
\end{theorem}

\begin{proof}
    Let $k$ be the capacity of $m$ facilities, hence $n=m k$ is the total number of agents.
    Toward a contradiction, let $M$ be a truthful and deterministic mechanism such that $ar_{MC}(M)=2-\delta$ where $\delta>0$.
    Let us consider the following instance: $x_1=0$ and $x_2=\dots=x_n=2$.
    It is easy to see that the optimal MC is $1$.
    Let $y_1$ denote the position at which $M$ places the facility to which agent at $x_1$ and $k-1$ of the agents at $2$ are assigned.
    Since $ar_{MC}(M)<2$, we have that $y_1\in[0,2]$.
    Given $t>0$, let us now consider the instance $x_1'=-t$ and $x_2=\dots=x_n=2$.
    For every $t>0$, the optimal MC of these instances is $\frac{t+2}{2}$.
    Since $M$ is truthful, we have that $y_1'\ge y_1$, thus $ar_{MC}(M)\ge \frac{t}{\frac{t+2}{2}}=\frac{2t}{t+2}$.
    Finally, we notice that the right hand-side of the inequality converges to $2$ as $t\to \infty$, thus we have that $ar(M)>2-\delta$ for every $\delta>0$, which is a contradiction.
\end{proof}

We now move to the lower bound for the Social Cost.

\begin{theorem}
\label{thm:lowerboundSC}
% There does not exists a truthful mechanism for the $m$-CFLP that achieves an approximation ratio lower than $3$ with respect to the Social Cost.
% 
No deterministic truthful mechanism for the $m$-CFLP with equi-capacitated facilities and no spare capacity can achieve an approximation ratio with respect to the Social Cost that is lower than $3$ whenever $k>3$.
\end{theorem}

\begin{proof}
    Given $m$ facilities with capacity $k$, let $n=mk$ be the number of agents.
    Let us consider the following instance: $x_1=\dots=x_{k+1}=0$ and $x_{k+2}=\dots=x_{n}=2$.
    It is easy to see that the optimal Social Cost is $2$.
    We now show that the mechanism must place at least one facility at $0$ and at least one facility at $2$.
    If all the facilities are placed at $0$, the approximation ratio of the mechanism would be higher than $k-1$, which would conclude the proof.
    Similarly, we conclude that not all the facilities are placed at $2$.
    Finally, let us assume that the facility serving the agents placed at $2$ and one of the agents placed at $0$, namely $y$ is such that $y\in(0,2)$.
    Without loss of generality, it suffices to consider the case in which one agent at $0$ shares the facility with agents at $2$, since in all other cases the cost of the mechanism increases.
    In this case, if we move the agent placed at $0$ to $y$, we have that the facility does not change its position (as otherwise an agent placed at $y$ could manipulate by reporting $0$), thus the approximation ratio of the mechanism would be at least equal to $k-1$.
    So the facility that serves both an agent at $0$ and $k-1$ agents at $2$ must be placed at $2$. 

    Let us consider one of the agents placed at $0$ that is not assigned to a facility placed at $2$.
    For every $\epsilon>0$, we have that if the agent was placed at $1-\epsilon$, it would be still assigned to a facility at $0$, as otherwise, it could manipulate by reporting $0$ rather than its real position.
    In this case, the cost of the mechanism is $3-\epsilon$, while the optimal cost is $1+\epsilon$.
    Since this holds for every $\epsilon>0$, the approximation ratio with respect to the SC of the mechanism is greater or equal to $3$.
\end{proof}

Finally, we present a lower bound for the approximation ratio with respect to the SC of deterministic, anonymous, and truthful mechanisms.
We recall that a mechanism M is anonymous if every agent's outcome depends only on its reports, i.e. two agents swapping two different reports causes the mechanism to swap their outcomes.
% (see the definition in \cite{ashlagi2009optimal}). 

\begin{theorem}
\label{thm_anonymLB}
    No deterministic, anonymous, and truthful mechanism for the $m$-CFLP with equi-capacitated facilities and no spare capacity can achieve an approximation ratio with respect to the Social Cost that is lower than $\big(\frac{k(m-1)}{2}+1\big)$ if $m$ is odd or lower than $\big(\frac{km}{2}-1\big)$ if $m$ is even. 
    % 
    % $2$.
    % There does not exists a truthful and anonymous mechanism whose approximation ratio with respect to the Social Cost is lower than $\big(k\floor{\frac{m}{2}}+1\big)$ if $m$ is odd or lower than $\big(k\frac{m}{2}-1\big)$ if $m$ is even.
\end{theorem}

\begin{proof}
    Let $M$ be a deterministic, anonymous, and truthful mechanism.
    Let us consider the following instance $x_1=\dots=x_{kr+1}=0$ and $x_{kr+2}=\dots=x_{n}=1$, where $r=\floor{\frac{m}{2}}$.
    The optimal cost of this instance is $1$.
    First, we show that, according to $M$, the locations of the facilities serving the agents at $0$ are all placed at the same distance from $0$.
    Toward a contradiction, let us assume that $M$ places two facilities at two positions, namely $y$ and $y'$, such that $|y-0|\neq |y'-0|$ and that both facilities serve an agent that reported $0$.
    Without loss of generality, let us assume that $|y-0|=|y|<|y'|=|y'-0|$.
    Let us denote with $x_i$ one of the agents who reported $0$ that is assigned to $y$.
    Let us now consider the instance $(y,x_{-i})$.
    Since $M$ is truthful, we must have that the mechanism places a facility at $y$ and that the agent at $y$ is assigned to it, as otherwise, it could misreport by reporting $0$.
    Let us now denote with $x_j$ one of the agents in $0$ that is assigned to $y'$.
    Since $M$ is anonymous, if $x_j$ reports $y$, it is assigned to $y$, which is closer to $0$ than $y'$, which contradicts the truthfulness of $M$.
    In particular, we infer that all the agents placed at $0$ incur the same cost.
    Similarly, all the agents placed at $1$ incur the same cost.
    Since there is no spare capacity, there exists at least one facility that serves an agent placed at $0$ and an agent placed at $1$, let us denote with $\lambda\in\erre$ its position on the line.
    Then, the total cost of the mechanism is $C=|\lambda|\big(k\floor{\frac{m}{2}}+1\big)+|1-\lambda|\big(n-k\floor{\frac{m}{2}}-1\big)$.
    % \[
    % C=|\lambda|\bigg(k\floor{\frac{m}{2}}+1\bigg)+|1-\lambda|\bigg(n-k\floor{\frac{m}{2}}-1\bigg).
    % \]
    Finally, we notice that $C\ge \big(\frac{k(m-1)}{2}+1\big)$ if $m$ is odd and $C\ge \big(\frac{km}{2}-1\big)$ if $m$ is even, which concludes the proof.
\end{proof}

Since the PMM and the PIPM are anonymous, the lower bound in Theorem \ref{thm_anonymLB} is tight.
Indeed PMM achieves the lower bound for odd $m$, while PIPM does so for even $m$.
Therefore, for the $m$-CFLP with equi-capacitated and no spare capacity, PMM and PIPM are the best anonymous, deterministic, and truthful mechanisms for odd and even $m$, respectively.
% 
% Lastly, we notice that since $mk=n$, we have that $ar_{SC}(M)\ge O(\frac{n}{2})$ for every truthful, anonymous and deterministic mechanism M.
% 
% In particular, when $m$ is even, the lower bound is always equal to $\frac{n}{2}-1$, regardless of the number of facilities and their capacity.

% \section{Worst Case Analysis Results for the $2$-CFLP}

\section{The \texorpdfstring{$2$}{2}-CFLP with abundant facilities}
\label{sec:particular_cases}

% 
% {\color{red} 
% 
We now consider the case in which we have to place two facilities capable of accommodating half of the agents.
% 
% We present the Extended InnerGap (EIG) mechanism, a truthful mechanism that generalizes and includes all the previously presented mechanisms (see Table Table \ref{tab:my_label}).
% 
We present the Extended InnerGap (EIG) mechanism, a truthful mechanism that generalizes and includes mechanisms that operate under further assumptions: the InnerPoint (IM) Mechanism \cite{aziz2020facility}, the InnerGap (IG) Mechanism \cite{ijcai2022p75}, and the InnerChoice (IC) Mechanism \cite{ijcai2022p75} (see Table \ref{tab:my_label}).
% 
% In particular, we unify three previously independents frameworks.
% 
We show that EIG achieves a finite approximation ratio with respect to the SC and the MC and corroborate these results by providing lower bounds on the approximation ratio achievable by truthful and deterministic mechanisms with respect to SC and MC.
As a consequence, we infer the approximation ratio of the IC and IG mechanisms, which, to the best of our knowledge, were previously unknown.

\begin{table}[]
    \centering
    \begin{tabular}{ @{}c c c c@{}}
    \toprule
          & $\forall n\in\mathbb{N}$ & $n<c_1+c_2$ &$c_1\neq c_2$ \\
          
         \midrule

          $EIG$ & Yes & Yes & Yes \\
          $EG$ & Yes & Yes & No \\
          $IC$ & No & No & Yes \\
          $IM$ & No & No & No \\
          \bottomrule
    \end{tabular}
    \caption{Frameworks under which the mechanisms operate when $c_1,c_2\ge \floor{\frac{n}{2}}$. From right to left, the column tell us whether the mechanism is capable of working ($1$) for every number of agents $n$, ($2$) when the total capacity is larger than the number of agents, and ($3$) when the two facilities have different capacities.
    The EIG (Exended InnerGap) Mechanism is the only mechanism capable of working under no further restriction.}
    \label{tab:my_label}
\end{table}
\begin{mechanism}[Extended InnerGap Mechanism]
\label{mech1}
% Given $c_1,c_2\ge \floor{\frac{n}{2}}$, set $\bar c:=\max\{c_1,c_2\}$.
% 
Let $\bar c:=\max\{c_1,c_2\}$ and let $\vec x\in\erre^n$ be the vector containing the agents' report ordered from left to right.
% , i.e. $x_i\le x_{i+1}$ for every $i\in[n-1]$.
% 
Let us fix $y_1=x_{n-\bar c}$, $y_2=x_{\bar c +1}$, and $z=\frac{y_1+y_2}{2}$, let $n_1$ be the number of agents in $[y_1,z]\cap \{x_i\}_{i\in [n]}$ and $n_2$ be the number of agents in $(z,y_2]\cap \{x_i\}_{i\in [n]}$.
% {\color{red}are agents reporting z counted in n1 or n2?}
% 
Finally, the output of the EIG over $\vec x$ is
\begin{enumerate*}[label=(\roman*)]
    \item to place the facility with the largest capacity at $y_1$ and the other at $y_2$ if $n_1\ge n_2$; or
    \item to place the facility with the lowest capacity at $y_1$ and the other at $y_2$ if $n_2> n_1$.
\end{enumerate*}
In both cases, every agent is assigned to the facility closer to its report.
\end{mechanism}
% }
% 
% Unlike the IG, the EIG does not require the two facilities to have equal capacity, but it works in every framework in which both facilities are able to accommodate at least half of the agents.
% % 
% In what follows, we show that the EIG is strong GSP, that it has bounded approximation ratio with respect to both SC and MC, and provide lower bounds to the approximation ratios of any truthful mechanism.
% 
% First, we study the truthfulness of the EIG.

\begin{theorem}
\label{thm:EIG_GSP}
    The EIG is strong GSP, hence truthful.
\end{theorem}

\begin{proof}
% [Proof of Theorem \ref{thm:EIG_GSP}]
% 
% {\color{red}
    Let $\vec x$ be the true positions of the agents.
    We denote with $y_1\le y_2$ the positions of the facilities according to the EIG on the truthful input.
    Let $I:=\{x_{i_1},\dots, x_{i_s}\}$ be the real positions of the agents that form a coalition able to manipulate the output of the EIG.
    Without loss of generality, we assume that $I$ is minimal, that is no subset of the agents in $I$ can collude.
    % 
    % In particular, we have that $x_i\neq x_i'$
    % 
    We recall that the EIG places the two facilities: one at the $(n-\bar c)$-th agents' report from the left, namely $y_1$, and one at the $(\bar c+1)$-th agents' report from the left, namely $y_2$.
    Since $I$ is minimal, none of the agents whose true position coincides with $y_1$ or $y_2$ takes part in the group manipulation.
    Hence, if we denote with $y_1'$ and $y_2'$ the positions of the facilities after the group manipulation, we cannot have that $y_1'< y_1$ and $y_2<y_2'$ at the same time.
    Let us now consider a coalition of agents $I$ that is able to lower the cost of an agent, whose real location is $x_{i_1}$, without increasing the cost of the other agents in $I$.
    Without loss of generality, let us assume that $x_{i_1}< y_1$, hence $y_1'<y_1$.
    If $y_1'<y_1$, it must be the case that at least one agent whose real position, namely $x_t$, was on the right of $y_1$ reports a position on the left of $y_1$, i.e. $x_t'\in I_1'$, where $x_t'$ is the misreport of the agent whose real position was $x_t$.
    If that agent was assigned to $y_2$ in the truthful input, it must be that $|x_t-y_2|\ge|x_t-y_1|>|x_t-y_1'|$, since $y_1'<y_1\le x_t$.
    Thus, the agent at $x_t$ is increasing its cost, which contradicts $x_t\in I$.
    Similarly, if $x_t$ was assigned to $y_1$ according to the truthful input, its cost still increases after the manipulation, which concludes the proof.
    % }
\end{proof}

The EIG mechanism determines the facility position using the same routine used by a percentile mechanism, \cite{sui2013analysis}.
However, the percentile mechanisms are not strong GSP in general, while the EIG mechanism is (see Example in the Appendix).
%
% \ref{ex:noGSP2} in the supplementary material).
% 
% 
% \footnote{The theorem proving the strong GSPness of percentile mechanisms in \cite{sui2013analysis} holds if all the agents have different positions, see Example \ref{ex:noGSP2} in the Appendix.}
% 
This difference is due to the fact that the EIG forces the agents to use a specific facility, while the percentile mechanism does not.
% 
% 
% We now show that both $ar_{SC}(\textrm{EIG})$ and $ar_{MC}(\textrm{EIG})$ are bounded.
% 
% While $ar_{SC}(\textrm{EIG})$ depends on the maximum capacity of the facilities, $ar_{MC}(\textrm{EIG})$ is constant.
% 

\begin{theorem}
\label{thm:EIGsc}
    % 
    % If $\bar c \le n-1$, i
    % 
    It holds that $ar_{SC}(EIG)=\max\{(n-\bar c-1),(\frac{\bar c}{n-\bar c}-1)\}$.
    Moreover, it holds that $ar_{MC}(EIG)=2$.
\end{theorem}

\begin{proof}
    We prove the statement only for the MC, the study of the SC is deferred to the Appendix.
    Let us denote with $I_i$ the set of agents that are assigned to the facility with capacity $c_i$ according to the optimal solution and, without loss of generality, we assume that all the agents in $I_1$ are placed to the left of the agents in $I_2$.
    The optimal MC is then $\frac{1}{2}\max\{|\min \{I_1\}-\max\{I_1\}|,|\min \{I_2\}-\max\{I_2\}|\}$.
    Let $y_1\le y_2$ be the position at which the mechanism places the two facilities.
    Then the MC of the EIG is lower or equal to the MC of assigning all the agents in $I_i$ to the facility at $y_i$.
    Finally, since $x_{n-\bar c}\in I_1$ and $x_{\bar c + 1}\in I_2$, we infer that
    \begin{align*}
        MC_{EIG}(\vec x)&\le \max\{\max_{x\in I_1}|x-x_{n-\bar c}|, \max_{x\in I_2}|x-x_{\bar c+1}|\}\\
        &\le\max\{|x_1-\max_{x\in I_1}\{x\}|,|\min_{x\in I_2}\{x\}-x_n|\} \\
        &\le 2 MC_{opt}(\vec x),
        % \max\{|\min \{I_1\}-\max\{I_1\}|,|\min \{I_2\}-\max\{I_2\}|\}\\
        % &= 2 MC_{opt}(\vec x).
    \end{align*}
    thus $ar_{MC}(EIG)\le 2$.
    Lastly, let us define $\vec x$ as $x_1=\dots=x_{\bar c +1}=0$, and $x_{\bar c+2}=\dots=x_{n}=1$.
    The optimal cost is $0.5$, while the cost of the EIG mechanism is $1$.  
\end{proof}

% {\color{red}
We now provide lower bounds on the approximation ratio with respect to both the MC and SC of any truthful and deterministic mechanism for this framework.
Our results show that the EIG is optimal or almost optimal for both costs.
\begin{theorem}
    \label{thm:lowerboundMCnodd2}
    Let $M$ be a truthful and deterministic mechanism that places two facilities with capacity $c_1,c_2\ge \floor{\frac{n}{2}}$, then we have that $ar_{MC}(M)\ge 2$ and $ar_{SC}(M)\ge 3$. 
    If $M$ is also anonymous, then we have that $ar_{SC}(M)\ge (n-\bar c-1)$.
\end{theorem}
% }
% % 
% {\color{red}
\begin{proof}
We prove only the lower bound with respect to the SC for truthful, deterministic, and anonymous mechanisms.
The proof for the other two cases, follow an argument similar to the ones used in the proof of Theorem \ref{thm:lowerboundMC} and \ref{thm:lowerboundSC} and are reported in the Appendix.
Let us consider the following instance: $x_1=\dots=x_{\bar c+1}=0$ and $x_{\bar c+2}=\dots=x_{n}=1$.
% 
% The lower bound on the approximation ratio of any deterministic and truthful mechanisms follows by the same argument used to prove Theorem \ref{thm_anonymLB}.
% 
By the same argument used in Theorem \ref{thm_anonymLB}, any truthful, deterministic, and anonymous mechanism places the two facilities at the same distance from $0$.
Since we have that $\bar c\ge \floor{\frac{n}{2}}$, we have that $\bar c+1\ge n-\bar c-1$, hence we get that the approximation ratio of any truthful, anonymous, and deterministic mechanism is larger than $(n-\bar c-1)$.
\end{proof}
% }
% {\color{red}
% 
In particular, the EIG is the best truthful, anonymous, and deterministic mechanism whenever $n\ge \bar c+\sqrt{\bar c}$.
% 
% }

\subsection{The EIG and previous mechanisms}

To conclude, we show that the EIG mechanisms extends and includes three already-known mechanisms.
In particular, 
\begin{enumerate*}[label=(\roman*)]
    \item when $n$ is an even number and $c_1=c_2=\frac{n}{2}$, the EIG mechanism coincides with the InnerPoint Mechanism, presented in \cite{aziz2020facility}.
    \item When $n=2k+1$ is odd, $c_1=k+1$, and $c_2=k$,  the EIG mechanism coincides with the InnerChoice Mechanism, presented in \cite{ijcai2022p75}.
    \item When $c_1=c_2$, the EIG mechanism coincides with the InnerGap Mechanism, presented in \cite{ijcai2022p75}.
\end{enumerate*}

For the sake of argument, we limit our discussion to the InnerChoice (IC) mechanism, and defer the other two cases to the Appendix.
Given an odd number $n=2k+1$ and two facilities whose capacities are $c_1=k+1$ and $c_2=k$, the routine of the IC mechanism is as follows:
\begin{enumerate*}[label=(\roman*)]
    \item Given $\vec x = (x_1, \dots, x_n)$ the vector containing the agents' reports ordered from left to right, i.e. $x_i\le x_{i+1}$, we define $\delta_1=|x_{k+1}-x_k|$ and $\delta_2=|x_{k+2}-x_{k+1}|$.     
    \item If $\delta_1\le \delta_2$, we locate the facility with capacity $c_1$ at $x_{k}$ and the other one at $x_{k+2}$. Otherwise, we locate the facility with capacity $c_1$ at $x_{k+2}$ and the other one at $x_{k}$.
    \item Lastly, every agent is assigned to its closest facility.
\end{enumerate*}
Since $\bar c=k+1$, we have that $x_{n-\bar c}=x_{k}$ and $x_{\bar c+1}=x_{k+2}$, hence, for every $\vec x\in\erre^n$, the output of EIG and IC are the same, thus the two mechanisms do coincide.
It was shown in \cite{ijcai2022p75} that the IC is truthful, however, owing to Theorem \ref{thm:EIG_GSP}, we have that IC is strong GSP.

\begin{theorem}
\label{thm:GSP_IC}
    The IC is strong Group Strategyproof.
\end{theorem}
Similarly, we extend the results on the approximation ratio of the IC mechanism with respect to the SC and MC.

\begin{theorem}
\label{thm:IMboundsSC}
    Let $n$ be an odd number, then $ar_{MC}(IC)=2$.
    Moreover, if $n>5$, it holds $ar_{SC}(IC)=k-1=\frac{n-3}{2}$, otherwise $ar_{SC}(IC)=1$.
\end{theorem}

Since $n\ge k+1+\sqrt{k+1}$, the IC is the optimal truthful, deterministic, and anonymous mechanisms to place two facilities of capacities $k+1$ and $k$ amongst $n=2k+1$ agents.
Moreover, the IC is also optimal with respect to the MC.

\begin{theorem}
    \label{thm:lowerboundMCnodd}
    Given $k\in\mathbb{N}$, let $n=2k+1>1$. 
    Then, every truthful deterministic mechanism $M$ that places two facilities with capacity $k+1$ and $k$ is such that $ar_{MC}(M)\ge 2$.
    Moreover, if $k>2$, $ar_{SC}(M)\ge 3$.
    Lastly, if $M$ is also anonymous, then $ar_{SC}(M)\ge k-1$.
\end{theorem}

Lastly, we notice that the only other mechanism known that is not extended by the EIG is the Extended Endpoint Mechanism (EEM).
However, the approximation ratio of EEM is equal to $\frac{3n}{2}$, which is larger than the one attained by the EIG, making it suboptimal \cite{aziz2020facility}.

\section{Conclusion and Future Works}

In this paper, we investigated two frameworks for the $m$-CFLP from a Mechanism Design perspective.
First, we considered the $m$-CFLP with equi-capacitated facilities and no spare capacity.
We propose two truthful mechanisms: the Propagating Median Mechanism (PMM) and the Propagating InnerPoint Mechanism (PIPM).
Both the mechanisms have bounded approximation ratios with respect to the Social and Maximum Costs.
We then established lower bounds on the approximation ratio of any truthful and deterministic mechanism for the $m$-CFLP with equi-capacitated facilities and no spare capacity. 
Notably, both PMM and PIPM achieved optimal approximation ratios for the Maximum Cost.
Additionally, we demonstrated that PMM and PIPM achieve the minimum possible approximation ratio for the Social Cost among truthful, deterministic, and anonymous mechanisms.
In the second framework, we considered the case in which we have two facilities to place and both facilities can accommodate half of the agents.
We proposed the Extended InnerGap mechanism, which is strong Group Strategyproof, achieves finite approximation ratio, is optimal with respect to the MC and almost optimal with respect to the SC.
In future research avenues, we aim to improve the lower bounds for non-anonymous mechanisms concerning the Social Cost, to explore higher-dimensional scenarios for agent placements, and to adapt existing randomized mechanisms to enhance approximation ratios results for this problem class \cite{procaccia2013approximate}.

% \appendix

% \section*{Ethical Statement}

% There are no ethical issues.

\section*{Acknowledgments}
Zihe Wang was partially supported by the National Natural Science Foundation of China (Grant No. 62172422). Jie Zhang was partially supported by a Leverhulme Trust Research Project Grant (2021 -- 2024) and the EPSRC grant (EP/W014912/1).

% The preparation of these instructions and the \LaTeX{} and Bib\TeX{}
% files that implement them was supported by Schlumberger Palo Alto
% Research, AT\&T Bell Laboratories, and Morgan Kaufmann Publishers.
% Preparation of the Microsoft Word file was supported by IJCAI.  An
% early version of this document was created by Shirley Jowell and Peter
% F. Patel-Schneider.  It was subsequently modified by Jennifer
% Ballentine, Thomas Dean, Bernhard Nebel, Daniel Pagenstecher,
% Kurt Steinkraus, Toby Walsh, Carles Sierra, Marc Pujol-Gonzalez,
% Francisco Cruz-Mencia and Edith Elkind.

%\newpage

\bibliographystyle{named}
\bibliography{ijcai24}

\clearpage
\section*{Appendix}

In this section, we report the missing proofs and the missing examples.

\subsection*{Missing Proofs}

\begin{proof}[Proof of Theorem \ref{thm:PMM_tr}]
    To conclude the proof, we need to consider the case in which $x_i\in I_j$, with $j>r$.
    We have two cases to analyze, depending on whether $y_j=x_{k(j-1)+1}$ or $y_j=x_{k\ell}+|y_{\ell}-x_{k\ell}|$, for an index $\ell<j$.
    If $y_j=x_{k(j-1)+1}$, we have that the facility is placed on the leftmost agent of $I_{j}$, thus this case is analogous to the case in which $x_i\in I_{r}$.
    Let us then consider the case in which $y_j=x_{k\ell}+|y_{\ell}-x_{k\ell}|$ for $\ell<j$.
    We break this case into two subcases: \begin{enumerate*}[label=(\roman*)]
        % \item $x_i$ 
        \item the manipulating agent is assigned to the $j$-th facility after manipulating and 
        % \item $x_i$ 
        \item the manipulating agent is assigned to another facility after it manipulates.
    \end{enumerate*}
    Let us consider the first subcase.
    % , i.e. $x_i$ is still assigned to the $j$-th facility after manipulating.
    % 
    By definition of PMM, there is only one case in which a single agent in $I_j$ can manipulate the position of $y_j$ while still being assigned to the $j$-th facility: the position of the manipulative agent is the leftmost position in $I_j$ and all the other positions in $I_j$ are on the right of $y_j$.
    In this case, however, the manipulative agent can only move the $j$-th facility further to the right, which increases its cost.
    Let us then consider the second subcase: the agent at $x_i$ manipulates in such a way that it is assigned to the $\ell'$-th facility, i.e. $y_{\ell'}$.
    Without loss of generality, let us assume that $\ell'<j$.
    If $\ell'\le r$, we have that $y_{\ell'}'\le y_r'\le y_r\le y_\ell$, thus the cost of the agent does not decrease since we have $|x_i-y_{\ell '}|\ge |x_i-y_r|\ge |x_i-y_j|$.
    If $r<\ell'\le \ell$, we have $x_t'=x_t$ for every $t\le k(\ell'-1)$.
    Since $y_{\ell '}'=\max\{x_{k(\ell'-1)+1}',x_{k(\ell'-1)}'+d_{(\ell'-1)}\}=\max\{x_{k(\ell'-1)+1}',x_{k(\ell'-1)}+d_{(\ell'-1)}\}$ and $x_{k(\ell-1)+1}'\le x_{k(\ell-1)+1}$, we have $y_{\ell'}'\le y_{\ell'}\le y_{\ell}$, thus
    % thus $y_{\ell}$ is closer to $x_i$ {\color{red}( or has the same distance to? )}than $y_{\ell'}'$, thus 
    $|x_i-y_{\ell'}'|\ge |x_i-y_{\ell}| \ge |x_i-y_j|$.
    Finally, if $\ell<\ell'<j$, we have  $x_t'=x_t$ for every $t\le k(\ell'-1)$, thus $y_{\ell'}=y_j=x_{k\ell}+|y_{\ell}-x_{k\ell}|=x_{k\ell}'+|y_{\ell}'-x_{k\ell}'|=y_{\ell'}'$ so the agent's cost is unchanged.
    % by definition of $y_j$, we have that $|x_i-y_j|\le |x_{(j-1)k}-y_j|=|x_{(j-1)k}-y_{j-1}|\le |x_{i}-y_{j-1}|$, since $x_i\in I_j$ and $x_i>x_{(j-1)k}$, which concludes the proof.}
\end{proof}

\begin{proof}[Proof of Theorem \ref{thm:PMM_MC}]
    First, we rewrite the optimal and the mechanism MC as the maximum of $m$ different costs.
    Indeed, we have $MC_{opt}(\vec x)=\max_{j\in[m]}\big\{MC_{opt}(I_j)\big\}$ and $MC_{\textrm{PMM}}(\vec x)=\max_{j\in[m]}\big\{MC_{\textrm{PMM}}(I_j)\big\}$, respectively, where $MC_{opt}(I_j)$ is the MC of the agents whose report is in $I_j$ according to the optimal solution and $MC_{\textrm{PMM}}(I_j)$ is the MC of the agents whose report is in $I_j$ according to the output of the PMM.
    Denoted with $\vec y$ the output of the PMM, we define $J$ as the set of indexes such that $y_j\in [x_{(j-1)k+1},x_{jk}]$.
    By definition of PMM, $r\in J$, thus $J$ is non-empty.
    If $j\in J$, using the same argument used to prove Theorem \ref{thm:PMM_SC}, we retrieve $MC_{\textrm{PMM}}(I_j)\le 2 MC_{opt}(I_j)$.
    If $j\notin J$ and $j>r$, there exists an index $\ell\in J$ such that $r\le\ell<j$ and $y_j=x_{k\ell}+|y_\ell-x_{k\ell}|$.
    For every  $x_t\in I_j$, we have $x_{k\ell}\le x_t$ and $|x_t-y_j|\le |x_{k\ell}-y_j|=|y_\ell-x_{k\ell}|\le MC_{\textrm{PMM}}(I_{r})$.
    Similarly, if $j\notin J$ and $j<r$, we conclude that there exists $j<\ell\le r$ such that $|x_t-y_j|\le MC_{\textrm{PMM}}(I_{\ell})$, for every $x_t\in I_j$.
    Then, we have $\max_{j\in[m]}MC_{\textrm{PMM}}(I_j)\le2\max_{j\in J} MC_{opt}(I_j)$ and $\max_{j\in J} MC_{opt}(I_j) \le\max_{j\in [m]} MC_{opt}(I_j) $, thus $ar_{MC}(PMM)\le 2$.
    % \[
    % \frac{\underset{j\in[m]}{\max}\big\{MC_{\textrm{PMM}}(I_j)\big\}}{\underset{j\in[m]}{\max}\big\{MC_{opt}(I_j)\big\}}\le\frac{\underset{j\in[m]}{\max}\big\{2MC_{opt}(I_j)\big\}}{\underset{j\in[m]}{\max}\big\{MC_{opt}(I_j)\big\}}\le 2.
    % \]
    % \begin{align*}
    % &\frac{\max_{j\in[m]}\big\{MC_{\textrm{PMM}}(I_j)\big\}}{\max_{j\in[m]}\big\{MC_{opt}(I_j)\big\}}\le\frac{\max_{j\in[m]}\big\{MC_{\textrm{PMM}}(I_j)\big\}}{\max_{j\in J}\big\{MC_{opt}(I_j)\big\}}\\
    % % &\qquad\qquad\qquad\qquad\quad\le\frac{\max_{j\in[m]}\big\{MC_{\textrm{PMM}}(I_j)\big\}}{\max_{j\in J}\big\{MC_{opt}(I_j)\big\}}\\
    % % &\le\frac{\max_{j\in J}\big\{M_{MC}(I_j)\big\}}{\max_{j\in J}\big\{OPT_{MC}(I_j)\big\}}\\
    % &\qquad\qquad\qquad\qquad\quad\le\frac{\max_{j\in J}\big\{2MC_{opt}(I_j)\big\}}{\max_{j\in J}\big\{MC_{opt}(I_j)\big\}}\le 2.
    % \end{align*}
    Hence the approximation ratio with respect to the MC is less than $2$.
    To prove that this bound is tight, consider the instance used in the proof of Theorem \ref{thm:PMM_SC}.
\end{proof}

\begin{proof}[Proof of Theorem \ref{thm:PIPM_all}]
We divide the proof into three pieces: in the first one we show that PIPM is truthful, in the second one we compute the approximation ratio of PIPM with respect to the Social Cost, and in the third one we compute the approximation ratio of PIPM with respect to the Maximum Cost.
\textbf{PIPM is truthful.}
    Toward a contradiction, let $x_i$ be the real position of an agent able to manipulate.
    We denote with $x_i'$ the position that the agent uses to manipulate the mechanism.
    We denote with $y_j$ the position of the facilities returned by the PIPM on the truthful input and with $y_j'$ the positions of the facilities returned by the PIPM when $x_i$ reports $x_i'$.
    Notice that the output of the PIPM is such that $y_1\le y_2\le\dots\le y_m$.
    Without loss of generality, we assume that $x_i\ge y_{r+1}=x_{r k+1}$, since the other case is symmetric.
    Finally, we recall that $I_j=\{x_{k(j-1)+1},x_{k(j-1)+2},\dots,x_{k(j-1)+k}\}$ and that, according to the PIPM, every agent in $I_j$ is assigned to the facility $j$-th facility.
    First, we show that if the agent at $x_i$ is able to manipulate, then $x_i\notin I_{r+1}$.
    Toward a contradiction, let us assume that $x_i\in I_{r+1}$.
    Then, if $x_i=y_{r+1}$, the cost of the agent is null, thus it cannot benefit by misreporting.
    Thus, it must be that $x_i>y_{r+1}$.
    If $y_{r+1}<x_i'<x_i$, the output of the mechanism does not change.
    If $x_i'<y_{r+1}$, we have that $y_r'\le y_{r+1}'<y_{r+1}$ and, since $x_i'\in I_\ell$ with $\ell\le r+1$, we have that the facility to which the manipulating agent is assigned is further to the left than $y_{r+1}$ so its costs is increased after the manipulation.
    % the cost in which $x_i$ incurs into.
    % 
    Finally, let us consider the case $x_i'>x_i$.
    In this case we have that $y'_{r+1}=y_{r+1}$, since the $(rk+1)$-th report from the left is the same regardless of whether the manipulating agent reports truthfully or not.
    Thus, if $x_i'\in I_{r+1}'$, it will still be assigned to $y_{r+1}'=y_{r+1}$, which brings no benefit to the manipulative agent.
    If $x_i'\notin I_{r+1}'$, then $x_i'\in I_\ell'$, where $\ell>r+1$, thus the manipulating agent is assigned to $y_{\ell}'\ge y_{r+2}'$, since $\ell>r+1$.
    Let us denote with $x_{(r+1)k}'$ the position of the $((r+1)k)$-th agent from the left in the manipulated instance $(x_i',x_{-i})$.
    Notice that, since $x_i'>x_i$, we have $x_{(r+1)k}'\ge x_{(r+1)k}\ge x_i$.
    By definition, we have that $y_{r+2}'=\max\{x_{(r+1)k+1}', x_{(r+1)k}'+|y_{r+1}'-x_{(r+1)k}'|\}\ge x_{(r+1)k}'+|y_{r+1}'-x_{(r+1)k}'| > x_i+|y_{r+1}-x_{(r+1)k}'| \ge x_i+|y_{r+1}-x_{(r+1)k}|$, hence the cost of being assigned to $y_\ell' $ is always greater or equal to the cost of being assigned to $y_{r+1}$.
    We then conclude that $x_i\notin I_{r+1}$.
    Lastly, let us consider the case in which $x_i\in I_j$, with $j>r+1$.
    We have two cases to analyze, depending on whether $y_j=x_{k(j-1)+1}$ or $y_j=x_{kr'}+|y_{r'}-x_{kr'}|$, with $r'<j$.
    If $y_j=x_{k(j-1)+1}$, we have that the facility is placed at the position of the leftmost agent in $I_{j}$, thus this case is analogous to the case in which $x_i\in I_{r+1}$.
    Finally, let us consider the case in which $y_j=x_{kr'}+|y_{r'}-x_{kr'}|$.
    We break this case into two subcases: \begin{enumerate*}[label=(\roman*)]
        \item the manipulative agent is still assigned to the $j$-th facility after manipulating, and
        \item the manipulative agent is assigned to another facility after it manipulates.
    \end{enumerate*}  
    Let us consider the first scenario.
    By definition of PIPM, there is only one case in which a single agent in $I_j$ can manipulate the position of $y_j$ while still being assigned to the $j$-th facility: the position of the manipulative agent is the leftmost position in $I_j$ and all the other agents in $I_j$ are on the right of $y_j$.
    In this case, however, manipulative agent can only move the $j$-th facility further to the right, which increases its cost.
    Let us then consider the second possible case: the manipulative agent reports in such a way that it is assigned to the $\ell$-th facility, i.e. $y_\ell'$.
    Without loss of generality, let us assume that $\ell<j$.
    Since the routine that determines the position of the facilities of the PIPM is the same used by the PMM, we can adapt the argument used to prove Theorem \ref{thm:PMM_SC} to infer that the agent is unable to lower its cost also in this case.
    % 

    % Again, we have that $y_{\ell}'\le y_{j-1}'\le y_{j-1}$, thus $y_{j-1}$ is closer to the real position of the manipulating agent than $y_{\ell}'$, thus $|x_i-y_\ell'|\ge |x_i-y_{j-1}|$.
    % % 
    % Finally, by definition of $y_j$, we have that $|x_i-y_j|\le |x_{(j-1)k}-y_j|=|x_{(j-1)k}-y_{j-1}|\le |x_{i}-y_{j-1}|$, since $x_i\in I_j$ and $x_i>x_{(j-1)k}$, which concludes the proof. 

    % If $\ell'\le r$, we have that $y_{\ell'}'\le y_r'\le y_r\le y_\ell$, thus the cost of the agent increases since we have $|x_i-y_{\ell '}|\le |x_i-y_r|\le |x_i-y_j|$.
    % % 
    % Let us now consider the case $r<\ell'\le \ell$, in this case, we have $x_t'=x_t$ for every $t\le k(\ell'-1)$.
    % % 
    % Since $y_{\ell '}'=\max\{x_{k(\ell'-1)+1}',x_{k(\ell'-1)}'+d_{(\ell'-1)}\}=\max\{x_{k(\ell'-1)+1}',x_{k(\ell'-1)}+d_{(\ell'-1)}\}$ and $x_{k(\ell-1)+1}'\le x_{k(\ell-1)+1}$, we have $y_{\ell'}'\le y_{\ell'}\le y_{\ell}$, thus $y_{\ell}$ is closer to $x_i$ than $y_{\ell'}'$, thus $|x_i-y_{\ell'}'|\ge |x_i-y_{j-1}|$.
    % % 
    % Finally, if $\ell<\ell'<j$, we have that $y_{\ell'}=y_j=x_{k\ell}+|y_{\ell}-x_{k\ell}|=x_{k\ell}'+|y_{\ell}'-x_{k\ell}'|=y_{\ell'}'$ since $x_t'=x_t$ for every $t\le k(\ell'-1)$, so the agent cannot benefit by performing this manipulation.

\textbf{The approximation ratio of the PIPM with respect to the Social Cost.}
    For the sake of simplicity, let us consider the case in which $m$ is even, so that $r=\frac{m}{2}$ is an integer.
    The case in which $m$ is odd is similar.
    Notice that, since $n=k m$, $n$ is also even.
    % 
    % We consider only the Social Cost and defer the Maximum Cost case to the appendix.
    % 
    Given instance $\vec x$, the optimal Social Cost $SC_{opt}(\vec x)$ and the cost of the mechanism $SC_{\textrm{PIPM}}(\vec x)$ is the sum of $m$ smaller costs, i.e. we have $SC_{opt}(\vec x)=\sum_{j=1}^m SC_{opt}(I_j)$ and $SC_{\textrm{PIPM}}(\vec x)=\sum_{j=1}^m SC_{\textrm{PIPM}}(I_j)$, where $SC_{opt}(I_j)$ is the Social Cost of the agents whose position is in $I_j$ according to the optimal solution, while $SC_{\textrm{PIPM}}(I_j)$ is the Social Cost of the agents whose position is in $I_j$ according to the output of the PIPM.
    Notice that, by definition of PIPM, agents whose position is in $I_j$ are assigned to the same facility.
    Let $\vec y=(y_1,\dots,y_m)$ be the positions at which the PIPM places the facilities when the input is $\vec x$.
    For every $j\in [m]$ we have two cases: either $y_j\in [x_{(j-1)k+1},x_{jk}]$ or $y_j\notin [x_{(j-1)k+1},x_{jk}]$.
    We denote with $J \subset [m]$ the set of indexes for which it holds $y_j\in [x_{(j-1)k+1},x_{jk}]$.
    Notice that, by definition of PIPM, we have $r,r+1\in J$, thus $J$ is always not empty.
    Let us assume that $j\in J$, so that $y_j\in [x_{(j-1)k+1},x_{jk}]$. 
    In this case, as previously shown in the proof of Theorem \ref{thm:PMM_SC}, we have that $SC_{\textrm{PIPM}}(I_j)\le (k-1)|x_{k(j-1)}+1-x_{kj}|\le (k-1)SC_{opt}(I_j)$.
    % 
    % Indeed, the cost of placing a facility between $x_{(j-1)k+1}$ and $x_{jk}$ that accommodates all the agents in $I_j$ is always less than $\max\big\{\sum_{l=1}^k|x_{(j-1)k+l}-x_{(j-1)k+1}|,\sum_{l=1}^k|x_{(j-1)k+l}-x_{jk}| \big\}$, which are the costs of the leftmost and rightmost mechanism, respectively.
    % % 
    % Since both the leftmost and the rightmost mechanisms have an approximation ratio equal to $k-1$, we conclude $SC_{\textrm{PIPM}}(I_j)\le (k-1)SC_{opt}(I_j)$.
    % 
    Let us now consider the case in which $y_j\notin [x_{(j-1)k+1},x_{jk}]$.
    By definition of $y_j$, we must have that $x_{jk}<y_j=x_{r'k}+|x_{r'k}-y_{r'}|$, where $r'<j$.
    Since $x_{r'k}\le x_{(j-1)k+1}\le \dots \le x_{jk}<y_j=x_{r'k}+|x_{r'k}-y_{r'}|$, we have that $|y_j-x_{(j-1)k+l}|<|x_{r'k}-y_{r'}|\le SC_{opt}(I_{r'})$ for every $l=1,2,\dots,k$.
    In particular, we have $SC_{\textrm{PIPM}}(I_j)\le k SC_{opt}(I_{r'})$.
    Therefore, we have that
    \begin{align}
    \label{eq:estimate_ar}
        \nonumber\frac{SC_{\textrm{PIPM}}(\vec x)}{SC_{opt}(\vec x)}&=\frac{\sum_{j=1}^m SC_{\textrm{PIPM}}(I_j)}{\sum_{j=1}^m SC_{opt}(I_j)}\\
        \nonumber&\le \frac{\sum_{j\in J}(t_jk+k-1) SC_{opt}(I_j)}{\sum_{j=1}^m SC_{opt}(I_j)}\\
        &\le \frac{\sum_{j\in J}(t_jk+k-1) SC_{opt}(I_j)}{\sum_{j\in J} SC_{opt}(I_j)}
    \end{align}
    where $J\subset [m]$ is the non-empty set of indexes for which it holds $y_j\in [x_{(j-1)k+1},x_{jk}]$, while $t_j$ is the number of sets $I_\ell$ such that $\ell\notin J$, $\ell> j$, and $\ell<j'$ for every $j'\in J$ such that $j<j'$.
    Let us now rewrite \eqref{eq:estimate_ar} as 
    \[
    \frac{SC_{\textrm{PIPM}}(\vec x)}{SC_{opt}(\vec x)}\le \sum_{j\in J}\beta_j\alpha_j,
    \]
    where $\beta_j=(t_jk+k-1)$ and $\alpha_j=\frac{SC_{opt}(I_j)}{\sum_{j\in J}SC_{opt}(I_j)}$.
    Since $\sum_{j\in J}\alpha_j=1$, we have that
    \[
    \frac{SC_{\textrm{PIPM}}(\vec x)}{SC_{opt}(\vec x)}\le \max_{j\in J}\;\beta_j.
    \]
    Finally, since $r,r+1\in J$, we have that the maximum possible value of $t_j$ is $r-1$, hence $\max\;\beta_j\le (r-1)k+k-1=\frac{m}{2}k-1=\frac{n}{2}-1$.
    We then conclude that the approximantion ratio of PIPM is less than or equal to $\frac{n}{2}-1$.
    To prove that the bound is tight, consider the following instance: $x_1=x_2=\dots = x_{\frac{n}{2}-1}=1$, $x_{\frac{n}{2}}=2$, $x_{\frac{n}{2}+1}=3$, and $x_{\frac{n}{2}+2}=x_{\frac{n}{2}+3}=\dots = x_{n}=4$.
    It is easy to see that the optimal cost of this instance is $2$.
    On this instance, the PIPM places $y_r$ at $2$, $y_{r+1}$ at $3$, all the facilities $y_\ell$ with $\ell<r$ at $0$, and all the other ones at $5$.
    The cost of the mechanism is then $2(\frac{n}{2}-1)$, thus the approximation ratio of the mechanism is greater or equal to $\frac{n}{2}-1$, which concludes the proof.

\textbf{The approximation ratio of the PIPM with respect to the Maximum Cost.}
    Let us now consider the Maximum Cost.
    Again, we rewrite the optimal and the mechanism Maximum Cost as the maximum of $m$ different costs, i.e. $MC_{opt}(\vec x)=\max_{j\in[m]}\big\{MC_{opt}(I_j)\big\}$ and $MC_{\textrm{PIPM}}(\vec x)=\max_{j\in[m]}\big\{MC_{\textrm{PIPM}}(I_j)\big\}$, respectively, where $MC_{opt}(I_j)$ is the Maximum Cost of the agents whose position is in $I_j$ according to the optimal solution and $MC_{\textrm{PIPM}}(I_j)$ is the Maximum Cost of the agents whose position is in $I_j$ according to the output of the mechanism.
    If $j\in J$, using the same argument used for the Social Cost, we retrieve $MC_{\textrm{PIPM}}(I_j)\le 2 MC_{opt}(I_j)$.
    If $j\notin J$, we can use the same argument used to prove Theorem \ref{thm:PMM_MC} to show that there exists an index $r\in J$ such that $r<j$ and $|x_\ell-y_j|\le MC_{\textrm{PIPM}}(I_{r})$ for every $x_\ell\in I_j$.
    Then, we have 
    \begin{align*}
    \frac{\max_{j\in[m]}\big\{MC_{\textrm{PIPM}}(I_j)\big\}}{\max_{j\in[m]}\big\{MC_{opt}(I_j)\big\}}&\le\frac{\max_{j\in[m]}\big\{MC_{\textrm{PIPM}}(I_j)\big\}}{\max_{j\in J}\big\{MC_{opt}(I_j)\big\}}\\
    &\le\frac{\max_{j\in J}\big\{MC_{\textrm{PIPM}}(I_j)\big\}}{\max_{j\in J}\big\{MC_{opt}(I_j)\big\}}\\
    % &\le\frac{\max_{j\in J}\big\{M_{MC}(I_j)\big\}}{\max_{j\in J}\big\{OPT_{MC}(I_j)\big\}}\\
    &\le\frac{\max_{j\in J}\big\{2MC_{opt}(I_j)\big\}}{\max_{j\in J}\big\{MC_{opt}(I_j)\big\}}\\
    &\le 2.
    \end{align*}
    Hence the approximation ratio with respect to the Maximum Cost is less than $2$.
    To prove that this bound is tight, it suffice to consider the same instance used for the Social Cost case.
\end{proof}

\begin{proof}[Proof of Theorem \ref{thm:EIG_GSP}]
% Since the EIG and the IC coincide when $n=2k+1$, $c_1=k+1$, and $c_2=k$, it follows directly from Theorem \ref{thm:EIG_GSP}.
% 
Let $\vec x$ be the vector containing the agents' real positions ordered from left to right.
    Let $y_1$ and $y_2$ be the positions at which the EIG mechanism places the two facilities on the truthful input.
Toward a contradiction, let $I:=\{x_{i_1},\dots, x_{i_s}\}$ be the real positions of a group of agents that can manipulate.
We denote with $I'=\{x'_{i_1},\dots,x'_{i_s}\}$ a group manipulation performed by $I$ such that \begin{enumerate*}[label=(\roman*)]
\item the cost of every agent in the group after the manipulation is less than or equal to the cost they would get by reporting truthfully,
\item at least one of costs of the agents in the group is strictly lower after the group manipulation.
\end{enumerate*}
We denote with $y'_1$ and $y_2'$ the positions at which IC places the facilities after the group manipulation.
Furthermore, assume that $I$ is minimal, i.e. there are no subsets of $I$ that can collude, thus $x_{i_l}\neq x_{i_l}'$ for every $l\in [s]$.
Notice that just swapping the capacities of the facilities without altering their position cannot bring any benefit to the agents.
We now show that no agent whose real position is $y_1=x_{n-\bar c}$ or $y_2=x_{\bar c + 1}$ are in $I$.
Indeed, if $x_{n-\bar c},x_{\bar c + 1}\in I$, then it must be that $y_1'=y_1$ and $y_2'=y_2$, as otherwise one of the two agents would increase its cost after manipulation, which is impossible.
If only one of the agents is in $I$, namely $x_{n-\bar c}=y_1$, then it must be that $y_1=y_1'$, thus it must be that $y_2'\neq y_2$.
If $y_2'<y_2$, we must have that none of the agents whose real position is to the right of $y_2$ takes part in the manipulation $I$, as otherwise their cost would strictly increase.
However, this is impossible as the $y_2$ facility is always placed at the position of the $(\bar c + 1)$-th agent from the left, which cannot be changed unless at least one agent whose real position is to the right of $y_2$ reports a value that is to the left of $y_2$.
Let us then consider the case in which $y_2'>y_2$.
Similarly, in order to alter the position of the $(\bar c + 1)$-th agent from the left, we must have that an agent on the left of $y_2$ reports a position that is on the right of $y_2$.
Hence, at least an agent whose real position is $x_i\le y_2$ reports a position $x_i'$ that is in $I'_2$ after the manipulation.
Again, this is a contradiction, since every agent whose position is in $I_2'$ is assigned to $y_2'>y_2\ge y_1$ and the agent's real position is to the left of $y_2$ hence $\min_{j=1,2}\{|x_i-y_j|\}<|x_i-y_2'|$.
Thus, the agents whose real position is $x_{n-\bar c}$ or $x_{\bar c + 1}$ cannot take part in the coalition.
% 
% Since the agents whose real position is either $x_{n-\bar c}$ and $x_{\bar c + 1}$ do not take part in the manipulation, we have that no group manipulation can ensure $y_1'<y_1$ and $y_2<y_2'$ at the same time.
% 

% 
% Hence, if we denote with $y_1'$ and $y_2'$ the positions of the facilities after the group manipulation, we cannot have that $y_1'< y_1$ and $y_2<y_2'$ at the same time.
% 
Let us now consider a coalition of agents $I$ that is able to lower the cost of an agent, whose real location is $x_{i_1}$, without increasing the cost of the other agents in $I$.
Without loss of generality, let us assume that $x_{i_1}< y_1$ so that it must be that $y_1'<y_1$.
If $y_1'<y_1$, it must be the case that at least one agent whose real position, namely $x_t$, was on the right of $y_1$ reports a position on the left of $y_1$, i.e. $x_t'\in I_1'$, where $x_t'$ is the misreport of the agent whose real position was $x_t$.
If that agent was assigned to $y_2$ in the truthful input, it must be that $|x_t-y_2|\le|x_t-y_1|<|x_t-y_1'|$, since $y_1'<y_1\le x_t$.
Thus, the agent at $x_t$ is increasing its cost, which contradicts the fact that $x_t$ takes part in the manipulating coalition $I$.
Similarly, if $x_t$ was assigned to $y_1$ according to the truthful input, its cost still increases after the manipulation, which concludes the proof.
\end{proof}

\begin{proof}[Proof of Theorem \ref{thm:EIGsc}]
\textbf{Approximation ratio for the SC.}
Let $\vec x$ be the vector containing the agents' reports ordered from left to right.
Since we are placing two facilities, the optimal solution splits the agents into two continuous sets, one served by the facility with capacity $c_1$ and the other served by the facility with capacity $c_2$.
We denote with $I_i$ the set of agents assigned to the facility with capacity $c_i$ according to the optimal solution. 
Without loss of generality, let us assume that $a\le b$ for every $a\in I_1$ and $b\in I_2$, thus if we denote with $y_i$ the position of the facility serving the agents in $I_i$, it holds $y_1\le y_2$.
Let us now denote with $y_1'\le y_2'$ the positions at which the EIG mechanism places the two facilities and let us denote with $I_i'$ the set of agents assigned to the facility located at $y_i'$.
By definition, all the agents in $I_1'$ are on the left of the agents in $I_2'$
It holds that $SC_{EIG}(\vec x)=SC_{EIG}(I_1')+SC_{EIG}(I_2')$, where $SC_{EIG}(I_i')$ is the SC of the agents in $I_i'$ according to the output of EIG.
Since EIG assigns the agents to their closest facility, we have that $SC_{EIG}(\vec x)\le SC_{EIG}(I_1)+SC_{EIG}(I_2)$, where $SC_{EIG}(I_i)$ is the Social Cost of the agents in $I_i$ if they are assigned to a facility placed at $y_i$ by the mechanism. 
Let $n_1$ be the number of agents in $I_1$.
Then, by definition of EIG, the value $SC_{EIG}(I_1)$ is the cost of a mechanism for the $1$-FLP that, given in input the reports of $n_1$ agents, locates the facility at the position of the $(n-\bar c)$-th agent to the left.
This mechanism has an approximation ratio equal to $AR_1\le\max\{\frac{n-\bar c-1}{n_1-(n-\bar c-1)},\frac{n_1-(n-\bar c)}{n-\bar c}\}$.
Indeed, given $m$ the position of the median agent in $\{x_1,\dots,x_{n_1}\}$, it is easy to see that the worst ratio between the SC of EIG and the optimal SC cost is the one in which all the agents are positioned at $m$ or at $x_{n-\bar c}$.
If $m=x_{n-\bar c}$, there is nothing to prove.
If $m\neq x_{n-\bar c}$, we have $SC_{opt}(x_1,\dots,x_{n_1})\ge \min\{n-\bar c,n_1-(n-\bar c-1)\}|x_{n-\bar c}-m|$.
Likewise, we have $SC_{EIG}(I_1)\le \max\{n-\bar c -1,n_1-(n-\bar c)\}|x_{n-\bar c}-m|$, hence 
\begin{align*}
    AR_1&\le\frac{\max\{n-\bar c -1,n_1-(n-\bar c)\}}{\min\{n-\bar c,n_1-(n-\bar c-1)\}}\\
    &=\max\Big\{\frac{n-\bar c-1}{n_1-(n-\bar c-1)},\frac{n_1-(n-\bar c)}{n-\bar c}\Big\}.
\end{align*}
Since $n \in\{n-\bar c,\bar c\}$, it holds $\frac{n-\bar c-1}{n_1-(n-\bar c-1)}\le n-\bar c -1$ and $\frac{n_1-(n-\bar c)}{n-\bar c}\le \frac{2\bar c - n}{n-\bar c}$, thus 
\[
SC_{EIG}(I_1)\le \max\Big\{\frac{2\bar c - n}{n-\bar c},n-\bar c-1\Big\} SC_{opt}(I_1).
\]
Through a similar argument, we infer that 
\[
SC_{EIG}(I_2)\le AR_2 \cdot SC_{opt}(I_2),
\]
where 
\[
AR_2=\max\Big\{\frac{\bar c-n_1}{(n-n_1)-(\bar c-n_1)},\frac{(n-n_1)-(\bar c-n_1+1)}{\bar c-n_1+1}\Big\}.
\]
Again, it is easy to see that $\frac{\bar c-n_1}{(n-n_1)-(\bar c-n_1)}=\frac{\bar c-n_1}{n-\bar c}\le\frac{2\bar c-n}{n-\bar c}$ and $\frac{(n-n_1)-(\bar c-n_1+1)}{\bar c-n_1+1}=\frac{n-\bar c-1}{\bar c-n_1+1}\le n-\bar c-1$, thus
\[
SC_{EIG}(I_2)\le \max\Big\{\frac{2\bar c - n}{n-\bar c},n-\bar c-1\Big\} SC_{opt}(I_2).
\]
Finally, since it holds $AR_1,AR_2\le \max\{n-\bar c -1,\frac{\bar c}{n-\bar c}-1\}$, we conclude that $ar(EIG)\le \max\{n-\bar c -1,\frac{\bar c}{n-\bar c}-1\}$.
Lastly, consider the following two instances.
In the first one we have $x_1=\dots=x_{n-\bar c -1}=0$, $x_{n-\bar c}=1$, and $x_i=5$ for every other $i\in [n]$.
The optimal solution has a cost equal to $1$, while EIG has a cost equal to $n-\bar c-1$, hence $ar(EIG)\ge n-\bar c-1$.
In the second instance we have $x_1=\dots=x_{n-\bar c }=0$, $x_{n-\bar c + 1}=\dots=x_{\bar c}=1$, and $x_i=2$ for all the other $i\in [n]$.
The optimal cost is $\min\{n-\bar c, \bar c-(n-\bar c)\}=\min\{n-\bar c, 2\bar c-n)\}$, while the mechanism cost is $ (2\bar c-n)$, thus we have that $ar(EIG)\ge \max\{1,\frac{2\bar c-n}{n-\bar c}\}=\max\{1,\frac{\bar c}{n-\bar c}-1\}$, which concludes the proof.

\end{proof}

\begin{proof}[Proof of Theorem \ref{thm:lowerboundMCnodd2}]
    \textbf{Lower bound with respect to the Maximum Cost.}
    It follows by the same argument used in the proof of Theorem \ref{thm:lowerboundMCnodd}.
    In this case however, the instance to consider is $x_1=\dots=x_{\bar c}=0$ and $x_{\bar c+1}=\dots=x_n=1$.
    \textbf{Lower bound with respect to the Social Cost.}
    Let us consider the following instance: $x_1=\dots=x_{\bar c+1}=0$ and $x_{\bar c+2}=\dots=x_{n}=1$.
    The lower bound on the approximation ratio of any deterministic and truthful mechanisms follows by the same argument used to prove Theorem \ref{thm:lowerboundSC}.
    By the same argument used in Theorem \ref{thm_anonymLB}, any truthful and anonymous mechanism places the two facilities at the same distance from $0$.
    Since we have that $\bar c\ge \floor{\frac{n}{2}}$, we have that $\bar c+1\ge n-\bar c-1$, hence we get that the approximation ratio of any truthful, anonymous, and deterministic mechanism is larger than $(n-\bar c-1)$.
\end{proof}

% \begin{proof}[Proof of Theorem \ref{thm:finalupperboundeig}]
%     Let us consider the following instance: $x_1=\dots=x_{\bar c+1}=0$ and $x_{\bar c+2}=\dots=x_{n}=1$.
%     % 
%     The lower bound on the approximation ratio of any deterministic and truthful mechanisms follows by the same argument used to prove Theorem \ref{}.
%     % 
%     By the same argument used in Theorem \ref{thm_anonymLB}, any truthful and anonymous mechanism places the two facilities at the same distance from $0$.
%     % 
%     Since we have that $\bar c\ge \floor{\frac{n}{2}}$, we have that $\bar c+1\ge n-\bar c-1$, hence we get that the approximation ratio of any truthful, anonymous, and deterministic mechanism is larger than $(n-\bar c-1)$.
% \end{proof}

% \begin{proof}[Proof of Corollary \ref{crll:samear}]
%     It follows from the fact that when $c_1=c_2$, both the IG and the EIG place the two facilities at the same position.
%     % 
%     Moreover, both the mechanisms assign the agents to the facility that is closer to the position reported by the agent.
%     % 
%     Since the cost of both mechanisms is the same on every instance, we conclude the proof.
% \end{proof}

\begin{proof}[Proof of Theorem \ref{thm:GSP_IC}]

Since the EIG mechanism and the IC mechanism do coincide when $n=2k+1$, $c_1=k+1$, and $c_2=k$, it follows directly from Theorem \ref{thm:EIG_GSP}.
\end{proof}

\begin{proof}[Proof of Theorem \ref{thm:IMboundsSC}]
% [Proof of Theorem \ref{thm:IMboundsMC}]
\textbf{Approximation ratio for the MC.}
    Without loss of generality, let us assume that the optimal solution splits the agents' positions as it follows $I_1=\{x_1,\dots,x_k\}$ and $I_2=\{x_{k+1},\dots,x_n\}$.
    Thus, the first $k$ agents are served by the facility with capacity $c_2=k$, while the remaining $k+1$ are served by the facility with capacity $c_1=k+1$. 
    The optimal location of the two facilities is $y_2=\frac{x_1+x_k}{2}$ and $y_1=\frac{x_{k+1}+x_n}{2}$, thus the optimal MC cost is $\frac{1}{2}\max\{|x_1-x_{k}|,|x_n-x_{k+1}|\}$.
    If the mechanism splits the agents optimally, i.e. the first $k$ agents are served by the facility placed at $x_k$ and the remaining $k+1$ by the facility placed at $x_{k+2}$, we have that the maximum cost of the mechanism is lower or equal to $\max\{|x_1-x_{k}|,|x_n-x_{k+1}|\}$.
    We then conclude that, on these instances the $ar_{MC}(IC)\le 2$.

    Let us now consider the case in which the optimal partition of agents does not coincide with the one returned by the mechanism.
    Again, let us assume that the optimal partition is still $I_1=\{x_1,\dots,x_k\}$ and $I_2=\{x_{k+1},\dots,x_n\}$, but, in this case, the partition returned by the mechanism is $I_1'=\{x_1,\dots,x_{k+1}\}$ and $I_2'=\{x_{k+2},\dots,x_n\}$.
    The cost of the mechanism is then $\max\{|x_1-x_k|,|x_{k+1}-x_k|,|x_{k+2}-x_n|\}$.
    Since the $(k+1)$-th agent is paired with $x_k$, it means that $|x_{k+1}-x_k|\le|x_{k+1}-x_{k+2}|$, thus, it holds
    \begin{align*}
        \max\{&|x_1-x_k|,|x_{k+1}-x_k|,|x_{k+2}-x_n|\}\\
        &\le\max\{|x_1-x_k|,|x_{k+1}-x_{k+2}|,|x_{k+2}-x_n|\}\\
        &\le\max\{|x_1-x_k|,|x_{k+1}-x_n|\}.
    \end{align*}
    % 
    % Finally, notice that the right hand-side is the cost of the mechanism if it partitioned the agents in the same way as the optimal solution.
    % 
    % From the previous point, w
    % 
    We then infer that the cost of the mechanism is again, at most double the optimal cost, i.e. $ar(IC)\le 2$.
    To prove that this bound is tight, we consider the following instance: let $x_1=\dots=x_k=0$, $x_{k+1}=\frac{1}{3}+\epsilon$, $x_{k+2}=\frac{2}{3}$, and $x_{k+3}=\dots=x_{n}=1$.
    For every $\epsilon>0$, IC locates the facility with capacity $c_2$ at $x_k=0$ and the facility with capacity $c_1$ at $x_{k+2}$, which leads to a Maximum Cost equal to $\frac{1}{3}$.
    However the optimal Maximum Cost is $\frac{1}{6}+\frac{\epsilon}{2}$, thus $ar(IC)\ge \frac{\frac{1}{3}}{\frac{1}{6}+\frac{\epsilon}{2}}=\frac{2}{1+3\epsilon}$, which converges to $2$ as $\epsilon$ goes to $0$.

    \textbf{Approximation ratio for the SC when $n>5$.}
    Let $\vec x=(x_1,\dots,x_n)$ be the agents' reports ordered from left to right, i.e. $x_1\le x_2\le \dots\le x_n$.
    Notice that every optimal solution splits the set of agents' positions into two subsets: the set containing the positions of the agents assigned to the facility with capacity $c_1$, namely $I_1$, and the set of the positions of the agents assigned to the facility with capacity $c_2$, namely $I_2$.
    Since the set of agents served by one facility is continuous and $c_1= c_2+1$, we have only two possibilities: $I_2=\{x_1,\dots,x_k\}$ and $I_1=\{x_{k+1},\dots,x_n\}$ or $I_1=\{x_1,\dots,x_{k+1}\}$ and $I_2=\{x_{k+2},\dots,x_n\}$.
    Since the other case is symmetric, we limit our discussion to the set of instances in which $I_1=\{x_1,\dots,x_{k+1}\}$ and $I_2=\{x_{k+2},\dots,x_n\}$.
    Let us now consider the solution found by the IC mechanism.
    As the optimal solution, the IC mechanism splits the set of the agents' positions into two sets, namely $I_1'$ and $I_2'$, respectively.
    Without loss of generality, let us assume that all the positions in $I_1'$ are on the left of all the positions in $I_2'$, thus the agents whose position is in $I_1'$ are assigned to the facility placed at $x_k$, while the agents whose position is in $I_2'$ are assigned to the facility at $x_{k+2}$.
    Since the IC assigns every agent to its closest facility, we have that 
    \begin{align*}
        SC_{IC}(\vec x)&=\sum_{x_i\in I_1'} |x_i-x_k|+\sum_{x_i\in I_2'}|x_i-x_{k+2}|\\
        &\le \sum_{x_i\in I_1} |x_i-x_k|+\sum_{x_i\in I_2}|x_i-x_{k+2}|\\
        &=  SC_{IC}(I_1)+SC_{IC}(I_2),
    \end{align*} 
    where $SC_{IC}(I_1)$ ($SC_{IC}(I_2)$) is the Social Cost of the agents whose position is in $I_1$ ($I_2$) if they are assigned to the facility placed at $x_k$ ($x_{k+2}$).
    % 
    % Thus, by definition of the IC, the agents in $I_1$ are assigned to a facility placed at $x_{k}$, while the agents in $I_2$  are assigned to a facility placed at $x_{k+2}$.
    % 
    Since $I_1=\{x_1,\dots,x_{k+1}\}$, $SC_{IC}(I_1)$ represents the Social Cost of a mechanism for the $1$-FLP that, given the position of $k+1$ agents, places a facility at the position of the second agent to the right.
    Since $k>2$, we have that the approximation ratio of such mechanism is $\frac{k-1}{2}$, so that $SC_{IC}(I_1)\le \frac{k-1}{2}SC_{opt}(I_1)$.
    Similarly, $SC_{IC}(I_2)$ is the cost of the leftmost mechanism over a set of $k$ agents, which has an approximation ratio equal to $k-1$, thus $SC_{IC}(I_2)\le (k-1)SC_{opt}(I_2)$.
    We then conclude that $ar_{SC}(IC)\le \frac{\frac{k-1}{2}SC_{opt}(I_1)+(k-1)SC_{opt}(I_2)}{SC_{opt}(I_1)+SC_{opt}(I_2)}$, since $\frac{k-1}{2}\le k-1$, we have $ar_{SC}(IC)\le k-1$.
    To conclude, consider the following instance: $x_1=\dots=x_{k+1}=0$, $x_{k+2}=1$, and $x_{k+3}=\dots=x_{n}=2$. 
    Indeed, the optimal Social Cost of this instance is $1$, while the cost of IC is $k-1$, thus $ar_{SC}(IC)=k-1$.
     \textbf{Approximation ratio with respect to the SC when $n=3,5$.}
    Finally, we consider the case $n=5$, the case $n=3$ is similar.
    Since $c_1=3$ and $c_2=2$, the optimal solution splits the agents' positions as $I_1=\{x_1,x_2,x_3\}$ and $I_2=\{x_4,x_5\}$ or as $I_1=\{x_1,x_2\}$ and $I_2=\{x_3,x_4,x_5\}$.
    Notice that, in both cases, placing the facilities at $x_2$ and $x_4$ would result in an optimal solution as long as we can correctly select where to place $c_1$ and $c_2$.
    Finally, we notice that if we place $c_1$ at $x_2$, the Social Cost of the solution is $|x_1-x_2|+|x_2-x_3|+|x_4-x_5|$, while if we place $c_1$ at $x_4$, the cost is $|x_1-x_2|+|x_3-x_4|+|x_4-x_5|$.
    It is then easy to see that the optimal solution places $c_1$ at $x_2$ if and only if $|x_2-x_3|\le |x_3-x_4|$.
    Since this is the routine that defines the IC mechanism, we conclude that the IC is optimal when $n=5$, $c_1=3$, and $c_2=2$.
\end{proof}

\begin{proof}[Proof of Theorem \ref{thm:lowerboundMCnodd}]
    It follows directly from Theorem \ref{thm:lowerboundMCnodd2}.
\end{proof}

\subsection*{Missing Examples}

\begin{example}
    \label{ex:noGSP2}
    Let us consider the following FLP problem with $5$ agents and $2$ facilities.
    We have that $x_1=0$, $x_2=1$, $x_3=x_4=2$, and $x_5=4$.
    Let us consider the percentile mechanism induced by the percentile vector $\vec p=(0.25,0.75)$.
    By definition of the percentile mechanism, we have that the facilities will be placed at the position of the $(\floor{(5-1)0.25}+1)$-th and $(\floor{(5-1)0.75}+1)$-th agents from the left, i.e. $x_2$ and $x_4$ in the truthful input.
    Notice that the $(0.25,0.75)$-percentile mechanism places the facilities at the same position as the EIG mechanism if both the facilities have capacity equal to $3$.
    However, if we use the percentile mechanism, the agent $x_4$ and $x_1$ can collude: indeed, if $x_4$ reports $0$ instead of $2$, the new input is $(0,0,1,2,4)$, thus the facilities are placed at $0$ and $2$ which reduces the cost of the agent at $0$ and leaves the cost of the agent at $2$ unchanged.
    Notice that, if we used the EIG mechanism to locate the facilities, the agent $x_4$ would have be forced to be assigned to the facility at $0$, which prevents the group manipulation.
\end{example}

\subsection*{The InnerGap Mechanism}

Given a number of agents $n$ and two facilities whose capacities are $c_1=c_2=k\ge \frac{n}{2}$, the routine of the IG mechanism is as follows:
\begin{enumerate*}[label=(\roman*)]
    \item Let $\vec x = (x_1, \dots, x_n)$ be the vector containing the agents' reports ordered from left to right, i.e. $x_i\le x_{i+1}$.     
    \item We set $y_1=x_{n-k}$ and $y_2=x_{k+1}$. 
    Since $c_1=c_2$, we do not need to specify the capacity of the facility placed at $y_1$ and $y_2$.
    \item Lastly, every agent is assigned to its closest facility.
\end{enumerate*}
Since $\bar c=c_1=c_2=k$, we have that $x_{n-\bar c}=x_{n-k}$ and $x_{\bar c+1}=x_{k+1}$, hence, for every $\vec x\in\erre^n$, the output of EIG and IG are the same, thus the two mechanisms do coincide.
It was shown in \cite{ijcai2022p75} that the IG is truthful, however, owing to Theorem \ref{thm:EIG_GSP}, we have that IG is strong GSP.

\begin{theorem}
    The IG is strong Group Strategyproof.
\end{theorem}

 \begin{proof}
     Since the routine of the IG is the same as the one of the EIG, the results follows from Theorem \ref{thm:EIG_GSP}.
 \end{proof}

Similarly, we extend the results on the approximation ratio of the IG mechanism with respect to the SC and MC.

\begin{theorem}
    Let $n$ be an odd number, then $ar_{MC}(IG)=2$.
    Moreover, it holds $ar_{SC}(IG)=\max\{(n-k-1),(\frac{k}{n-k}-1)\}$.
\end{theorem}

\begin{proof}
     Since the routine of the IG is the same as the one of the EIG, the results follows from Theorem \ref{thm:EIGsc}.
 \end{proof}

\begin{theorem}
    Let $c_1=c_2=k\in\mathbb{N}$ and $n\in\mathbb{N}$. 
    Then, every truthful deterministic mechanism $M$ that places two facilities with capacity $c_1=c_2=k$ is such that $ar_{MC}(M)\ge 2$.
    Likewise, any truthful and deterministic mechanism $M$ is such that $ar_{SC}(M)\ge 3$.
    Moreover, if $M$ is also anonymous, then $ar_{SC}(M)\ge n-k-1$.
\end{theorem}
\begin{proof}
     Since the routine of the IG is the same as the one of the EIG, the results follows from Theorem \ref{thm:lowerboundMCnodd2}.
 \end{proof}

\subsection*{The Innerpoint Mechanism.}

Given an even number $n=2k$ and two facilities whose capacities are $c_1=c_2=k$, the routine of the Innerpoint Mechanism (IM) is as follows:
\begin{enumerate*}[label=(\roman*)]
    \item Given $\vec x = (x_1, \dots, x_n)$ the vector containing the agents' reports ordered from left to right.     
    \item We set $y_1=x_{k}$ and $y_2=x_{k+1}$. 
    Since $c_1=c_2$, we do not need to specify the capacity of the facility placed at $y_1$ and $y_2$.
    \item Lastly, every agent is assigned to its closest facility.
\end{enumerate*}
It is easy to see that the IM is the IG when $c_1=c_2=\frac{n}{2}$, thus all the results presented for the EIG and IG do apply to this case.

\end{document}